\documentclass[11pt]{article}
\usepackage{amsmath}
\usepackage{fancyhdr}
\usepackage{amssymb}
\usepackage{amsthm}
\usepackage{graphicx}
\usepackage{varioref}
\usepackage{verbatim} 
\usepackage{multicol}
\usepackage{lmodern}
\usepackage{enumerate}
\usepackage[normalem]{ulem}
\usepackage{caption}
\usepackage{subcaption}
\usepackage[T1]{fontenc}
\usepackage[margin=1in]{geometry}
\usepackage{fancyhdr}
\usepackage{authblk}
\usepackage{tikz}
\usetikzlibrary{patterns}
\usetikzlibrary{arrows}

\usepackage{mathrsfs}

\usepackage{url}
\usepackage{xspace}
\usepackage{thm-restate}

\usepackage{hyperref}
\hypersetup{
    bookmarksnumbered=true, 
    unicode=false, 
    pdfstartview={}, 
    pdftitle={}, 
    pdfauthor={}, 
    pdfsubject={}, 
    pdfcreator={}, 
    pdfproducer={}, 
    pdfkeywords={}, 
    pdfnewwindow=true, 
    colorlinks=true, 
    linkcolor=blue, 
    citecolor=blue, 
    filecolor=blue, 
    urlcolor=blue 
}

\newtheorem{theorem}{Theorem}
\newtheorem{definition}{Definition}
\newtheorem{lemma}{Lemma}

\newcommand{\ket}[1]{|#1\rangle}
\newcommand{\bra}[1]{\langle#1|}

\DeclareMathAlphabet{\matheu}{U}{eus}{m}{n}

\newcommand{\sop}[1]{{\mathcal #1}}

\newcommand{\eps}{\varepsilon}

\newcommand{\norm}[1]{\left\| #1 \right\|}

\def\tO{\widetilde{\mathrm{O}}}

\renewcommand{\(}{\left(}
\renewcommand{\)}{\right)}

\begin{document}

\title{\textsc{nand}-Trees, Average Choice Complexity, and Effective Resistance}
\author{Stacey Jeffery\thanks{Institute for Quantum Information and Matter (IQIM), Caltech, \texttt{sjeffery@caltech.edu}} and Shelby Kimmel\thanks{Joint Center for Quantum Information and Computer Science (QuICS),
University of Maryland, \texttt{shelbyk@umd.edu}}}

\date{}

\maketitle

\begin{abstract}

We show that the quantum query complexity of evaluating \textsc{nand}-tree
instances with \emph{average choice complexity} at most $W$ is $O(W)$,
where average choice complexity is a measure of the difficulty of
winning the associated two-player game. This generalizes a
superpolynomial speedup over classical query complexity due to Zhan
\emph{et al}. We further show that the player with a winning strategy
for the two-player game associated with the \textsc{nand}-tree can win
the game with an expected $\widetilde{O}(N^{1/4}\sqrt{{\cal C}(x)})$
quantum queries against a random opponent, where ${\cal C }(x)$ is the
average choice complexity of the instance. This gives an improvement
over the query complexity of the naive strategy, which costs
$\widetilde{O}(\sqrt{N})$ queries.

The results rely on a connection between \textsc{nand}-tree evaluation
and $st$-connectivity problems on certain graphs, and span programs
for $st$-connectivity problems. Our results follow from relating
average choice complexity to the effective resistance of these graphs,
which itself corresponds to the span program witness size.

\end{abstract}

\section{Introduction}

\textsc{nand}-tree evaluation is a Boolean formula evaluation problem that
has proven to be a rich ground for developing
quantum algorithms. The first quantum algorithm for evaluating \textsc{nand}-trees
was created by Farhi, Goldstone, and Gutmann in the continuous time
model \cite{FGG07}, and showed that a tree with input size $N$, corresponding to a formula consisting of nested \textsc{nand}-gates to depth $\log N$,
could be evaluated in time $O(\sqrt{N})$. They used an algorithm
that involved scattering a wavepacket off of a graph.  
The \textsc{nand}-tree problem
was quickly adapted to the query model, \cite{CCJ+09,ACR+10}, and
seems to be a major motivation for the development of span
program quantum algorithms in \cite{RS12}. Through the
refinement of span program algorithms, Reichardt was able to 
show that any formula with $O(1)$ fan-in on $N$ variables
(of which \textsc{nand}-trees are a special case) can be evaluated in 
$O(\sqrt{N})$ queries \cite{R14}. Classically,
the query complexity of evaluating \textsc{nand}-trees is $\Theta(N^{.753})$ \cite{SW86}.
Variants of 
the \textsc{nand}-tree problem have also been studied. For example, when the
input is in a certain class, called \emph{$k$-fault trees}, the
quantum query complexity can be improved to $O(2^k)$ \cite{ZKH12,K11}. 
A lower bound of $\Omega((\log \log N-\log k)^k)$ on the classical query 
complexity of evaluating $k$-fault trees
makes this a \emph{superpolynomial} quantum speedup for
a range of values of~$k$~\cite{ZKH12}.

Every \textsc{nand}-tree instance also defines a two-player game as
follows. Players $A$ and $B$ move on a full binary tree of depth $\log
N$, with leaves labeled by the input bits of $x\in\{0,1\}^{N}$.
Starting from the root, $A$ and $B$ alternate choosing one child
of the current node to move to until a leaf is reached. If the value
of the leaf is $1$, $A$ wins, and if the value of the leaf is $0$, $B$
wins. It turns out that for $1$-instances of \textsc{nand}-tree
evaluation, $A$ can always win if she plays optimally, no matter what
strategy $B$ employs, and for $0$-instances, $B$ can always win if he
plays optimally.

The results of \cite{ZKH12} show a connection between the difficulty
of winning the game associated with a \textsc{nand}-tree instance, as
measured by the \emph{fault complexity}, and the quantum query
complexity of evaluating a \textsc{nand}-tree instance. A
\textsc{nand}-tree instance has fault complexity $2^k$ if the player
with the winning strategy will encounter at most $k$ \emph{faults} on
any winning path, where a fault is a node at which one choice leads to
a sub-game in which the player still has a winning strategy, while the
other does not.  In this paper, we consider a more nuanced measure of
the difficulty of winning the two-player game associated with a
\textsc{nand}-tree, which we call the \emph{average choice
complexity}. The average choice
complexity is upper bounded by the fault complexity. Like
the fault complexity, the average choice complexity is related to the
number of critical decisions a player must make to win the two-player
game, but rather than describing the worst case, as with the fault
measure, the average choice complexity depends on decisions made in a
good strategy, averaged over the opponent's strategy. The average
choice complexity also provides a more subtle characterization of the
criticality of the choice made at a particular node, and captures the
fact that nodes that are not faults can still be important decision
points.

By exploiting an elegant connection between \textsc{nand}-trees (or
$(\wedge,\vee)$-formulas in general) and $st$-connectivity problems on
certain graphs, we find that the average choice complexity is exactly
the span program witness size in a span program for evaluating
\textsc{nand}-trees. We are thus able to generalize the
superpolynomial speedup of \cite{ZKH12} to show that the bounded error
quantum query complexity of evaluating \textsc{nand}-trees, when we
are promised that the average choice complexity of the input is at
most $W$, is $O(W)$ (Theorem \ref{thm:easy-instances}).

A related problem to \emph{evaluating} a \textsc{nand}-tree instance
is \emph{winning} a \textsc{nand}-tree instance.  The connection
between \textsc{nand}-tree evaluation and winning strategies suggests
a strategy for winning the game: at every node, solve the instances of
\textsc{nand}-tree rooted at each child, and if possible, always
choose one that evaluates to 1 if you are Player $A$, or 0 if you are
Player $B$. The total number of quantum queries employed by this
strategy is $\widetilde{O}(\sqrt{N})$. As a second application of the
connection between average choice complexity and span program witness
size, we are able to give a strategy that wins a \textsc{nand}-tree
against a random opponent using an expected
$\widetilde{O}(N^{1/4}\sqrt{\mathcal{C}})$ quantum queries, where
$\cal C$ is the average choice complexity (Theorem \ref{thm:winning}).
This result gives an operational interpretation of average choice
complexity as the difficulty, in quantum query complexity, of winning
against a random opponent. The algorithm uses a recent result of
\cite{IJ15} that constructs an algorithm for estimating the span
program witness size of any span program. The player uses
this estimation algorithm to
estimate the average choice complexity of both children and then
chooses the path that is easier on average.

We achieve the connection between average choice complexity and
witness size by exploiting a link between evaluating
\textsc{nand}-trees and solving $st$-connectivity in certain graphs,
which may be of further interest. In particular, we find that 
 $1$-valued \textsc{nand}-tree instances are related to $st$-connected subgraphs of a certain 
 graph $G$, while
$0$-valued instances are related to $st$-connected subgraphs of the dual of $G$. 

The relationship we uncover between \textsc{nand}-trees,
$st$-connectivity problems on graphs and their duals, and two-player games hints at
many connections that might be explored. When is the quantum query complexity of evaluating Boolean formulas characterized
by graph connectivity problems? What role do graphs and their duals play in 
span program algorithms? Our current span program algorithm characterizes
properties of a two-player game assuming the opponent plays randomly; 
by adjusting the span program algorithm,
could we characterize the same two-player game, but with different
strategies for the players? On the other hand, it may be that these strategies 
are somehow particularly natural for quantum algorithms, in which case,
we would like to understand why.

\paragraph{Outline} In Section \ref{sec:Preliminaries}, we give some
requisite background information. In Section  
\ref{sec:average-choice-complexity}, we define average choice complexity. In Section
\ref{sec:eval}, we prove our first main result: that
\textsc{nand}-trees with average choice complexity at most $W$ can be
evaluated in $O(W)$ quantum queries. We also show in this section that  the average
choice complexity can be estimated using
$\tO\(\sqrt{\sop C(x)} N^{1/4}\)$ queries, where $\sop C(x)$ is
the average choice complexity of the instance. We prove these results by exploiting a
connection between \textsc{nand}-trees and $st$-connectivity on a
family of graphs and their duals. We first relate the average choice
complexity to the effective resistance of these graphs (Section
\ref{sec:nand-graphs}). Then, we construct and analyze a span program for
$st$-connectivity problems on these graphs, such that the witness size
of the span program depends on the effective resistance of the graph (Section
\ref{sec:span}). In Section \ref{sec:win_nand_tree}, we give
a quantum algorithm for playing the two-player \textsc{nand}-tree
game that makes use of our algorithm for estimating ${\cal C}(x)$.
Finally, in Section \ref{sec:open_problems}, we discuss some open problems.

\section{Preliminaries}\label{sec:Preliminaries}

In this section, we will provide the necessary information to understand our paper:
in Section \ref{sec:intro_NAND}, we describe \textsc{nand}-trees, in Section
\ref{sec:intro_span} we provide some basic results about span programs
and quantum algorithms, and in Section \ref{sec:intro_graph} we introduce
some key concepts from graph theory.

\subsection{NAND-Trees}\label{sec:intro_NAND}

A \textsc{nand}-tree is a full binary tree of depth $\ell$ in
which each leaf is labeled by a variable $x_i$, and each internal
node is labeled by either $\vee$ (\textsc{or}), if it is at even distance from the
leaves, or $\wedge$ (\textsc{and}), if it is at odd
distance from the leaves. While a \textsc{nand}-tree of
depth $\ell$ is sometimes defined as a Boolean formula of \textsc{nand}s composed to
depth $\ell$, we will instead think of the formula as an alternation of
\textsc{and}s and \textsc{or}s --- when $\ell$ is even, these two characterizations are identical. An instance of \textsc{nand}-tree is a binary
string $x\in \{0,1\}^{N}$, where $N=2^{\ell}$.  We use $\textsc{nand}_\ell$ to denote a complete \textsc{nand}-tree
of depth $\ell$. For instance, a \textsc{nand}-tree of depth $2$ represents
the formula:
\begin{equation}
\textsc{nand}_2(x_1,x_2,x_3,x_4)
=(x_1\wedge x_2)\vee (x_3\wedge x_4).
\end{equation}
If $\textsc{nand}_{\ell}(x)=1$, we say that $x$ is a
$1$-instance, and otherwise we say it is a 0-instance.

A \textsc{nand}-tree instance can be associated with a two-player game. Let $x\in\{0,1\}^{N}$ be an input to a
depth-$\ell$ \textsc{nand}-tree, so $N=2^{\ell}$. Then we associate this
input to a game played on a full binary tree as in Figure \ref{fig:nandtree},
where the leaves take the values $x_i$. The game starts at the root
node $v$, which we call the live node. If the live node is at
even distance from the leaves (an \textsc{or} node)  Player $A$ chooses to move to
one of the live node's children. The chosen child now becomes the live node. At each further round, as
long as the live node is not a leaf, if the live node is at
even (respectively odd) distance from the leaves, Player $A$ (resp.\ Player $B$) chooses one
of the live node's children to
become the live node. When the live node becomes a leaf, if the leaf
has value $1$, then Player $A$ wins, and if the leaf has value $0$, then
Player $B$ wins. The sequence of moves by each player determines a path
from the root to a leaf.

A simple inductive argument shows that if $x$ is a 1-instance of \textsc{nand}-tree, then there exists a strategy by which Player
$A$ can always win, no matter what strategy $B$ employs; and if $x$ is a $0$-instance, there exists a strategy by which Player $B$ can always win. 
We say an input $x$ is $A$-winnable if it has value 1 and
$B$-winnable if it has value $0$.

In this paper, we will consider trees of even depth for simplicity, but the arguments can be easily
extended to odd depth trees. 

\subsection{Span Programs}\label{sec:intro_span}

In this section, we review the concept of span programs, and their use in quantum algorithms. Span programs \cite{KW93} were first introduced to the study of quantum algorithms by Reichardt and \v{S}palek \cite{RS12}. They have since proven to be a tool of immense importance for designing quantum algorithms in the query model. 

\begin{definition}[Span Program]\label{def:span}
A span program $P=(H,V,\tau,A)$ on $\{0,1\}^n$ is made up of 
\begin{enumerate}
\item finite-dimensional inner product spaces $H=H_1\oplus \dots \oplus H_n$, and $\{H_{j,b}\subseteq H_j\}_{j\in [n],b\in \{0,1\}}$ such that $H_{j,0}+H_{j,1}=H_j$,
\item a vector space $V$,
\item a \emph{target vector} $\tau\in V$, and
\item a linear operator $A:H\rightarrow V$.
\end{enumerate}
For every string $x\in \{0,1\}^n$, we associate the subspace $H(x):=H_{1,x_1}\oplus \dots\oplus H_{n,x_n}$, and an operator $A(x):=A\Pi_{H(x)}$, where $\Pi_{H(x)}$ is the orthogonal projector onto $H(x)$. 
\end{definition}

Given a span program, one can obtain a quantum 
query algorithm based on that span program. The parameters
of the span program that determine the query complexity of
the algorithm are the positive and negative witness sizes:

\begin{definition}[Positive and Negative Witness]
Let $P$ be a span program on $\{0,1\}^n$ and let $x$ be a string $x\in \{0,1\}^n$.
Then we call $\ket{w}$ a \emph{positive witness for $x$ in $P$} if
$\ket{w}\in H(x)$, and $A\ket{w}=\tau$. We define the \emph{positive
witness size of $x$} as:
$$w_+(x,P)=w_+(x)=\min\{\norm{\ket{w}}^2: \ket{w}\in H(x):A\ket{w}=\tau\},$$
if there exists a positive witness for $x$, and $w_+(x)=\infty$ otherwise.
We call a linear map $\omega:V\rightarrow \mathbb{R}$ a \emph{negative
witness for $x$ in $P$} if $\omega A\Pi_{H(x)} = 0$ and $\omega\tau =
1$. We define the \emph{negative witness size of $x$} as:
$$w_-(x,P)=w_-(x)=\min\{\norm{\omega A}^2:{\omega\in \mathcal{L}(V,\mathbb{R}): 
\omega A\Pi_{H(x)}=0, \omega\tau=1}\},$$
if there exists a negative witness, and $w_-(x)=\infty$ otherwise.
If $w_+(x)$ is finite, we say that $x$ is \emph{positive} (wrt.\ $P$), 
and if $w_-(x)$ is finite, we say that $x$ is \emph{negative}. We let 
$P_1$ denote the set of positive inputs, and $P_0$ the set of negative 
inputs for $P$.  
\end{definition}

For a decision problem $f:X\rightarrow \{0,1\}$, with $X\subseteq\{0,1\}^n$, we say
that $P$ \emph{decides} $f$ if $f^{-1}(0)\subseteq P_0$ and
$f^{-1}(1)\subseteq P_1$. In that case, we can use $P$ to construct a
quantum algorithm that decides $f$, where the number of queries used
by the algorithm depends on the witness sizes:

\begin{theorem}[\cite{Rei09}]\label{thm:span-decision}
Fix $X\subseteq\{0,1\}^n$ and $f:X\rightarrow\{0,1\}$, and let $P$ be a span program on $\{0,1\}^n$ that decides $f$. 
Let $W_+(f,P)=\max_{x\in f^{-1}(1)}w_+(x,P)$ 
and $W_-(f,P)=\max_{x\in f^{-1}(0)}w_-(x,P)$. 
Then there is a bounded error quantum algorithm that decides $f$ with quantum query complexity $O(\sqrt{W_+(f,P)W_-(f,P)})$.
\end{theorem}

In \cite{IJ15}, Ito and Jeffery show that additional quantum algorithms
can be derived from a span program $P$. These algorithms depend on the
approximate positive and negative witness sizes. 

\begin{definition}[Approximate Positive Witness]
For any span program $P$ on $\{0,1\}^n$ and $x\in\{0,1\}^n$, we define the
\emph{positive error of $x$ in $P$} as: 
$$e_+(x)=e_+(x,P):=\min\left\{
\norm{\Pi_{H(x)^\bot}\ket{w}}^2:A\ket{w}=\tau\right\}.$$ 
We say $\ket{w}$ is an \emph{approximate positive witness} for
$x$ in $P$ if $\norm{\Pi_{H(x)^\bot}\ket{w}}^2=e_+(x)$ and $A\ket{w}=\tau.$
 We define the \emph{approximate positive witness size} as
$$\tilde{w}_+(x)=\tilde{w}_+(x,P):=\min\left\{\norm{\ket{w}}^2:A\ket{w}=\tau,
\norm{\Pi_{H(x)^\bot}\ket{w}}^2=e_+(x)\right\}.$$ 
\end{definition}

\noindent  Note that if $x\in P_1$, then $e_+(x)=0$. In that case, an
approximate positive witness for $x$ is a positive witness, and
$\tilde w_+(x)=w_+(x)$. For negative inputs,
the positive error is larger than 0. 

We can define a similar notion of approximate negative witnesses: 

\begin{definition}[Approximate Negative Witness]
For any span program $P$ on $\{0,1\}^n$ and $x\in\{0,1\}^n$,
we define the \emph{negative error of $x$ in $P$} as:
$$e_-(x)=e_-(x,P):=\min\left\{\norm{\omega A\Pi_{H(x)}}^2: \omega(\tau)=1\right\}.$$
Any $\omega$ such that $\norm{\omega A\Pi_{H(x)}}^2=e_-(x,P)$ is called an 
\emph{approximate negative witness} for $x$ in $P$. We define 
the \emph{approximate negative witness size} as
$$\tilde{w}_-(x)=\tilde{w}_-(x,P):=\min\left\{\norm{\omega A}^2:
\omega(\tau)=1,\norm{\omega A\Pi_{H(x)}}^2=e_-(x,P)\right\}.$$
\end{definition}

\noindent  Note that if $x\in P_0$, then $e_-(x)=0$. In that case, an
approximate negative witness for $x$ is a negative witness, and
$\tilde w_-(x)=w_-(x)$. For positive inputs,
the negative error is larger than 0.  

Ito and Jeffery give a quantum algorithm to estimate the positive
witness size (resp.\ negative witness size) of a span program; the query complexity of the algorithm
depends on the approximate negative (resp.\ positive) witness sizes:

\begin{theorem}[Witness Size Estimation Algorithm, \cite{IJ15}]\label{thm:span-est}
Fix $X\subseteq\{0,1\}^n$ and $f:X\rightarrow \mathbb{R}_{\geq 0}$. Let $P$ be a
span program such that for all $x\in X$, $f(x)=w_+(x,P)$ and define
$\widetilde{W}_-=\widetilde{W}_-(f,P)=\max_{x\in X}\tilde{w}_-(x,P)$.
There exists a quantum algorithm that estimates $f$ to relative error $\eps$
and that uses $\tO\(\frac{1}{\eps^{3/2}}\sqrt{w_+(x)\widetilde{W}_-}\)$ queries. 
Similarly, let $P$ be a span program such that for all $x\in X$,
$f(x)=w_-(x,P)$ and define
$\widetilde{W}_+=\widetilde{W}_+(f,P)=\max_{x\in X}\tilde{w}_+(x,P)$.
Then there exists a quantum algorithm that estimates $f$ to accuracy
$\eps$ and that uses $\tO\(\frac{1}{\eps^{3/2}}\sqrt{w_-(x)\widetilde{W}_+}\)$
queries.
\end{theorem}

\subsection{Graph Theory}\label{sec:intro_graph}

 For an undirected graph $G$, we let
$V(G)$ denote the vertices of $G$, and $E(G)$ denote the edges of $G$. We
also define $\overrightarrow{E}(G)=\{(u,v):\{u,v\}\in E(G)\}$
We will consider the effective resistance of graphs, for which we need the concept of a unit flow:
\begin{definition}[Unit Flow] 
Let $G$ be an undirected graph with $s,t\in V(G)$. Then a \emph{unit $st$-flow} on $G$ is a function $\theta:\overrightarrow{E}(G)\rightarrow\mathbb{R}$ such that:
\begin{enumerate}
\item For all $(u,v)\in \overrightarrow{E}(G)$, $\theta(u,v)=-\theta(v,u)$;
\item $\sum_{v\in \Gamma(s)}\theta(s,v)=\sum_{v\in\Gamma(t)}\theta(v,t)=1$, where $\Gamma(u)$ is the neighbourhood of $u$ in $G$; and 
\item for all $u\in V(G)\setminus\{s,t\}$, $\sum_{v\in\Gamma(u)}\theta(u,v)=0$. 
\end{enumerate}
\end{definition}

\noindent A \emph{circulation} is a function $\theta$ that satisfies (1), and in addition, satisifes (3) for \emph{all} vertices in $V(G)$. 

\begin{definition}[Unit Flow Energy]
Given a unit $st$-flow $\theta$ on a graph $G$,  
the \emph{unit flow energy} is 
$J(\theta)=\sum_{\{u,v\}\in {E}}\theta(u,v)^2.$
\end{definition}

\noindent The effective resistance between $s$ and $t$ is the smallest energy 
of any unit $st$-flow:
\begin{definition}[Effective Resistance]
Fix an undirected graph $G$, and $s,t\in V(G)$. 
The effective resistance is defined $R_{s,t}(G)=\infty$ 
if $s$ and $t$ are not connected in $G$, and otherwise, 
$R_{s,t}(G)=\min_\theta\sum_{\{u,v\}\in {E}}\theta(u,v)^2$, 
where $\theta$ runs over all unit $st$-flows in $G$. 
\end{definition}

\noindent We will also look at dual graphs:
\begin{definition}[Dual Graph]

Let $G$ be a planar graph (with an implicit planar embedding). The \emph{dual} graph, $G^\dagger$, is
defined as follows. For every face $f$ of $G$, $G^\dagger$ has a
vertex $v_f$, and any two vertices are adjacent if their corresponding
faces share an edge, $e$. We call the edge between two such vertices the
\emph{dual edge} to $e$, $e^\dagger$.

When $G$ is a directed graph, we assign edge directions to the dual
 by orienting each dual edge $\pi/2$ radians clockwise from the primal edge. 
\end{definition}

\section{Average Choice Complexity}\label{sec:average-choice-complexity}

In \cite{ZKH12}, Zhan, Kimmel and Hassidim find a relationship between
the difficulty of playing the two-player game associated with a
\textsc{nand}-tree, and the witness size. Specifically, they find that trees with a smaller number of \emph{faults}, or critical decisions for a player playing the associated two-player game, are easier to evaluate on a quantum computer. In this section, we give a more nuanced measure of the difficulty of a two-player game, the \emph{average choice complexity}, which we will later see is also connected with the quantum query complexity of evaluating a \textsc{nand}-tree. 

Zhan \emph{et al}.\ measure the difficulty of a tree in terms of
faults. When playing a game, a fault is a node where if the player
chooses the correct direction, she can win the game if she plays
optimally, but if the player chooses the wrong direction, she is no
longer guaranteed to win the game, and in fact will lose if the
opposing player plays optimally.  Equivalently, it is a node in which one child corresponds to a 0-instance of \textsc{nand}-tree, and the other child corresponds to a 1-instance. Let $x$ be a $Z$-winnable input for
$Z\in\{A,B\}.$ Consider all paths such that Player $Z$ wins, and also
Player $Z$ never makes a move that would allow her opponent to win. We call
such paths \emph{$Z$-winning paths}.
Zhan \emph{et al}.\ call a tree  $k$-fault\footnote{We have actually
used the more refined definition of $k$-fault from \cite{K11}} if each
of these paths encounters at most $k$ faults. Then the witness size of
such a tree is less than $2^k$ \cite{K11}, and we call $2^k$ the 
\emph{fault complexity} of the tree. Fault complexity
encapsulates the worst case number of critical decisions that must be made in
order to win the game, on any winning path.

In our case, the average choice complexity, defined shortly, does not characterize the number
of decisions in the worst case over any winning path, but is rather a characterization of
the expected number of decisions required when playing against a random opponent. 

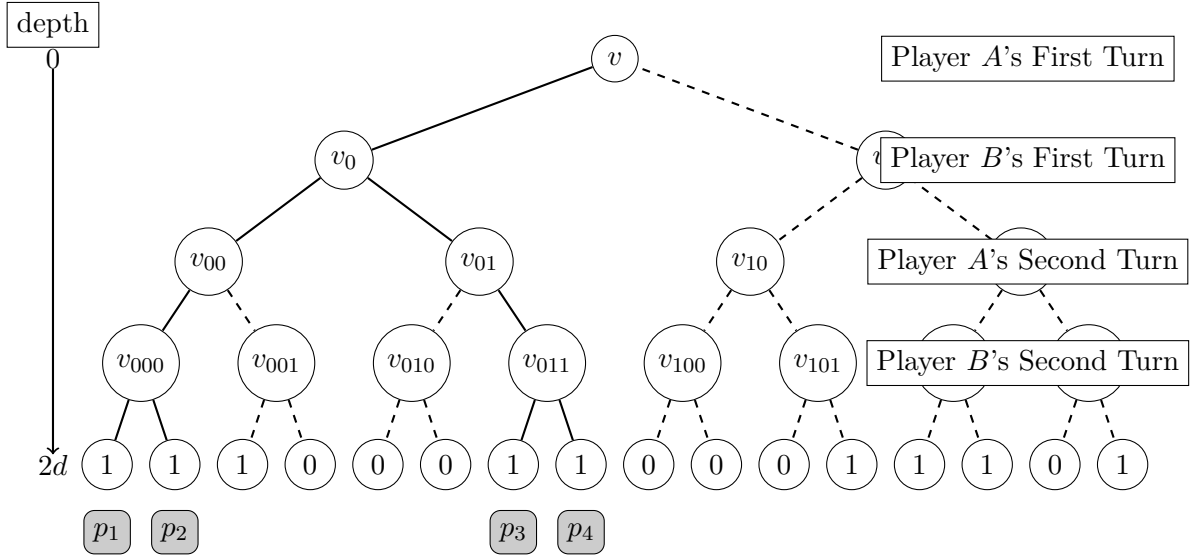
\begin{figure}[h]
\centering
\begin{tikzpicture}[scale=.9]
\tikzstyle{vertex} = [circle,draw,fill=white,minimum size=.5em]
\tikzstyle{operator} = [rectangle,rounded corners,draw,
fill = black!20!white,
minimum size=1.5em]
\tikzstyle{turn} = [rectangle,draw,fill=white,minimum size=1.5em]
\node[vertex] (v0000) at (-7.5,0) {$1$};
\node[operator] (p1) at (-7.5,-1) {$p_1$};
\node[vertex] (v0001) at (-6.5,0) {$1$};
\node[operator] (p2) at (-6.5,-1) {$p_2$};
\node[vertex] (v0010) at (-5.5,0) {$1$};
\node[vertex] (v0011) at (-4.5,0) {$0$};
\node[vertex] (v0100) at (-3.5,0) {$0$};
\node[vertex] (v0101) at (-2.5,0) {$0$};
\node[vertex] (v0110) at (-1.5,0) {$1$};
\node[operator] (p3) at (-1.5,-1) {$p_3$};
\node[vertex] (v0111) at (-0.5,0) {$1$};
\node[operator] (p4) at (-0.5,-1) {$p_4$};
\node[vertex] (v1000) at (0.5,0) {$0$};
\node[vertex] (v1001) at (1.5,0) {$0$};
\node[vertex] (v1010) at (2.5,0) {$0$};
\node[vertex] (v1011) at (3.5,0) {$1$};
\node[vertex] (v1100) at (4.5,0) {$1$};
\node[vertex] (v1101) at (5.5,0) {$1$};
\node[vertex] (v1110) at (6.5,0) {$0$};
\node[vertex] (v1111) at (7.5,0) {$1$};
\node[vertex] (v000) at (-7,1.5) {$v_{000}$};
\node[vertex] (v001) at (-5,1.5) {$v_{001}$};
\node[vertex] (v010) at (-3,1.5) {$v_{010}$};
\node[vertex] (v011) at (-1,1.5) {$v_{011}$};
\node[vertex] (v100) at (1,1.5) {$v_{100}$};
\node[vertex] (v101) at (3,1.5) {$v_{101}$};
\node[vertex] (v110) at (5,1.5) {$v_{110}$};
\node[vertex] (v111) at (7,1.5) {$v_{111}$};
\node[vertex] (v00) at (-6,3) {$v_{00}$};
\node[vertex] (v01) at (-2,3) {$v_{01}$};
\node[vertex] (v10) at (2,3) {$v_{10}$};
\node[vertex] (v11) at (6,3) {$v_{11}$};
\node[vertex] (v0) at (-4,4.5) {$v_{0}$};
\node[vertex] (v1) at (4,4.5) {$v_{1}$};
\node[vertex] (v) at (0,6) {$v$};
\path (v000) edge[thick] (v0001);
\path (v000) edge[thick] (v0000);
\path[dashed] (v00) edge[thick] (v001);
\path (v00) edge[thick] (v000);
\path (v0) edge[thick] (v01);
\path (v0) edge[thick] (v00);
\path[dashed] (v001) edge[thick] (v0011);
\path[dashed] (v001) edge[thick] (v0010);
\path (v01) edge[thick] (v011);
\path[dashed] (v01) edge[thick] (v010);
\path[dashed] (v1) edge[thick] (v11);
\path (v011) edge[thick] (v0111);
\path (v011) edge[thick] (v0110);
\path[dashed] (v111) edge[thick] (v1110);
\path[dashed] (v111) edge[thick] (v1111);
\path[dashed] (v010) edge[thick] (v0101);
\path[dashed] (v010) edge[thick] (v0100);
\path[dashed] (v100) edge[thick] (v1001);
\path[dashed] (v100) edge[thick] (v1000);
\path[dashed] (v101) edge[thick] (v1011);
\path[dashed] (v101) edge[thick] (v1010);
\path[dashed] (v110) edge[thick] (v1101);
\path[dashed] (v110) edge[thick] (v1100);
\path[dashed] (v10) edge[thick] (v101);
\path[dashed] (v10) edge[thick] (v100);
\path[dashed] (v1) edge[thick] (v10);
\path[dashed] (v11) edge[thick] (v111);
\path[dashed] (v11) edge[thick] (v110);
\path[dashed] (v1) edge[thick] (v);
\path (v0) edge[thick] (v);
\node[turn] (p1m) at (6.1,6) {$\textrm{Player $A$'s First Turn}$};
\node[turn] (p1m) at (6.1,4.5) {$\textrm{Player $B$'s First Turn}$};
\node[turn] (p1m) at (6.1,3) {$\textrm{Player $A$'s Second Turn}$};
\node[turn] (p1m) at (6.1,1.5) {$\textrm{Player $B$'s Second Turn}$};
\node (h1) at (-8.3,0) {};
\node (h2) at (-8.3,6) {};
\path (h2) edge[thick,->] (h1);
\node[turn] (height) at (-8.3,6.5) {depth}; 
\node (zh) at (-8.3,0) {$2d$};
\node (zb) at (-8.3,6) {$0$};
\end{tikzpicture}
\caption{A winning strategy for Player $A$ is defined by the
solid lines. There are four possible paths that the
winning strategy might take, labeled by $p_i$.}
\label{fig:nandtree}
\end{figure}

Consider an even-depth \textsc{nand}-tree instance $x\in\{0,1\}^{2^{2d}}$. A $Z$-strategy $\sop S(x)$ for input $x$  is a complete deterministic
proscription for how Player $Z$ should act given any action of the
opposing player.  Any strategy consists of $2^d$ paths from the root to a leaf, which
describe all possible responses to the opposing player. A
$Z$-{\it{winning}} strategy is a $Z$-strategy $\sop S$ such that all paths are $Z$-winning paths.  An example of a winning strategy for Player $A$ and
the paths of the strategy are given in Figure \ref{fig:nandtree}.  In
a valid strategy, only one choice at each of Player $Z$'s turns must be
on a path. For example, in Figure \ref{fig:nandtree}, at node $v$, to
follow any of the paths, Player $A$ must choose to go to node $v_0$
rather than~$v_1$. Along any given path $p$, let $\nu_Z(p)$ be the set
of nodes along the path  at which it is Player $Z$'s turn. Thus, $\nu_A(p)$ contains those nodes in $p$ at even distance $>0$ from the leaf, and $\nu_B(p)$ contains those nodes at odd distance from the leaf.

To each node $v$, we assign a criticality parameter $c(v)$ that ranges between $1$
and $2$, and is a representation of how important that decision is. A
value of $1$ signals that the decision made by the player at this node is inconsequential: loosely speaking, there are just as many ways to win in either sub-game the player could choose. However a
criticality parameter of $2$ signals that the decision is critical: one of the sub-games is winnable, and the other is not. In other  words, any node with criticality value 2 is a
fault node.

\begin{definition}[Node Criticality and Average Choice Complexity]\label{def:av_choice}

For $x\in\{0,1\}^{2^\ell}$ for $\ell\geq 0$ (note $\ell$ can be even or odd), let ${\cal C}_A(x)$ be the average choice complexity for Player $A$
on input $x$, and ${\cal C}_B(x)$ be the average choice complexity for Player $B$.
 If $x$ is $A$-winnable, then ${\cal C}_B(x)=\infty$, and if
$x$ is $B$-winnable, then ${\cal C}_A(x)=\infty.$ For a depth-$0$ tree, we define ${\cal C}_A(1)=1$ and ${\cal C}_B(0)=1$. 
For depth $\ell$ trees with $\ell>0$, we define average choice complexity, ${\cal C}_Z$,
for $Z\in\{A,B\}$
recursively in terms of the \emph{criticality}, $c(v)$, defined shortly:
\begin{equation} \label{eq:Av_comp_def}
{\cal C}_Z(x) = \left\{\begin{array}{ll}
\min_{{\cal S}\in\mathscr{S}_Z(x)}\mathbb{E}_{p\in {\cal S}}\left[\prod_{v\in \nu_Z(p)}c(v) \right] & \mbox{if $x$ is $Z$-winnable}\\
\infty & \mbox{else,}
\end{array}\right.
\end{equation}
where $\mathscr{S}_Z(x)$ is the set of all $Z$-winning strategies. 

We now define $c(v)$. As a base case, let $v$ be a node whose children are leaves --- that is, it is the root of a depth-1 subtree --- with children labeled $x^0,x^1\in\{0,1\}$. Then define:
\begin{equation}
c(v)=\left\{\begin{array}{ll}
1 & \mbox{if }x^0=x^1\\
2 & \mbox{if }x^0\neq x^1.
\end{array}\right.
\end{equation}
Let $v$ be a node that is the root of a depth-$\ell$ subtree for $\ell>1$. If $\ell$ is even, fix $Z'=A$, and if $\ell$ is odd, fix $Z'=B$.
If $v$ is the root of a tree that is not $Z'$-winnable, then we define $c(v)=1$. If $v$ is a fault, then we define $c(v)=2$. 
Otherwise, let $x^{00},x^{01},x^{10},$ and $x^{11}$ be the respective inputs to the subtrees rooted at
$v_{00},v_{01},v_{10},$ and $v_{11}$, the four grandchildren of $v$. Then ${\cal C}_{Z'}(x^{00}),{\cal C}_{Z'}(x^{01}),{\cal C}_{Z'}(x^{10}),{\cal C}_{Z'}(x^{11})$ are all finite, and we define:
\begin{equation}
c(v)=\frac{\max\left\{{\cal C}_{Z'}(x^{00})+{\cal C}_{Z'}(x^{01}), {\cal C}_{Z'}(x^{10})+{\cal C}_{Z'}(x^{11})\right\}}{\mathrm{avg}\left\{{\cal C}_{Z'}(x^{00})+{\cal C}_{Z'}(x^{01}), {\cal C}_{Z'}(x^{10})+{\cal C}_{Z'}(x^{11})\right\}}.
\end{equation}

\noindent Note that exactly one of ${\cal C}_A(x)$ and ${\cal C}_B(x)$ is finite. Finally, we define:
\begin{equation}
{\cal C}(x)=\min\{{\cal C}_A(x),{\cal C}_B(x)\}.
\end{equation}
\end{definition}

The criticality
should be thought of as measuring the difference between making a
random choice at node $v$, and making a good choice at node $v$, with
$c(v)=1$ indicating that the two choices at $v$ are equal. The expressions 
${\cal C}_{Z'}(x^{00})+{\cal C}_{Z'}(x^{01})$ and ${\cal C}_{Z'}(x^{00})+{\cal C}_{Z'}(x^{01})$ are proportional to the average complexity 
(averaged over the opposing player's next decision) faced by Player $Z'$ in each of the respective paths she might take. If these two values are nearly the same, it does not matter so much which path Player $Z'$ takes at this turn, and $c(v)$ is close to $1$. If these two values are very different, then $c(v)$ approaches $2$. 
The criticality parameter is always
in the range $[1,2].$

In Section \ref{sec:win_nand_tree}, we further motivate the average choice complexity by showing that it is related to the quantum query complexity of winning a two-player \textsc{nand}-tree game. 
We now compare the average choice complexity to fault complexity.
If we let
$\bar{c}(v)$ assign a value of $2$ to every fault node, and 1 to every
non-fault, and take the max over all $Z$-winning paths $p$, we recover the fault
complexity of Zhan \emph{et al}.:

\begin{equation}
{\cal F}_Z(x)=\max_{p\in P}\prod_{v\in \nu_Z(p)}\bar{c}(v)=\max_{p\in P}2^{k_p},
\end{equation}
where $x$ is $Z$-winnable, $P$ runs over all $Z$-winning paths, and $k_p$ is the number of faults in $\nu_Z(p)$. 
When $x$ is not Z-winnable, ${\cal F}_Z(x)=\infty.$ Then the fault complexity is
${\cal F}(x)=\min\{{\cal F}_A(x),{\cal F}_B(x)\}.$

We prove in Lemma \ref{lemma:kfault_conn} that ${\cal C}(x)\leq {\cal F}(x)$. 
In contrast to the setting of fault complexity, Definition~\ref{def:av_choice} allows the criticality parameter
to range between 1 and 2, only considers paths in a certain strategy, and takes an average
over those paths rather than looking at the maximum. In the worst case,
when every winning path has exactly $k$ faults, and every non-fault node on every winning path
has criticality $1$, we have ${\cal C}(x)={\cal F}(x)$, but in general, ${\cal C}(x)$ may be significantly smaller. 


\section{Evaluating NAND-Trees with low Average Choice Complexity}\label{sec:eval}

In this section, we generalize the speedup of Zhan \emph{et al}.\ by proving the following theorem.
\begin{theorem}\label{thm:easy-instances}
Fix $W=W(d)$. For all $d$, let $X_d\subseteq \{0,1\}^{2^{2d}}$ be the set
of instances such that ${\cal C}(x)\leq W$.
Then the bounded error quantum query complexity of evaluating $\textsc{nand}$-trees in $X=\bigcup_dX_d$ is $O(W)$.
\end{theorem}

Let $k$ be the largest integer such that $2^k\leq W$. Then by Lemma \ref{lemma:kfault_conn}, $X$
contains the set of $k$-fault trees. Thus, the lower bound of
$\Omega((\log\log N-\log k)^k)$ on the classical query complexity of
$k$-fault trees from \cite{ZKH12} implies a lower bound of
$\Omega((\log\log N-\log\log W)^{\log W})$ on evaluating \textsc{nand}-trees in
$X$, making Theorem \ref{thm:easy-instances} a superpolynomial
speedup as well.


To prove Theorem \ref{thm:easy-instances}, we describe in Section \ref{sec:nand-graphs} how \textsc{nand}-tree evaluation is equivalent to solving an $st$-connectivity problem on certain graphs, $G_d(x)$. In Lemma \ref{lemma:kfault_conn}, we show that ${\cal C}_A(x)=R_{s,t}(G_d(x))$ and ${\cal C}_B(x)$ is the effective resistance on the dual-complement of $G_d(x)$. In Section \ref{sec:span}, we present a span program whose positive witness sizes correspond to $R_{s,t}(G_d(x))$ (Lemma \ref{lemma:pos_witness}), and negative witness sizes correspond to the effective resistance on the dual complement (Lemma \ref{lemma:negative_witness}), completing the proof. 

\vskip10pt

\noindent In this section, we also prove the following theorem, which will be used in Section \ref{sec:win_nand_tree}:

\begin{theorem}\label{thm:est}
Let $Z\in\{A,B\}$. The bounded error quantum query complexity of estimating the average choice complexity for Player $Z$ 
of a \textsc{nand}-tree instance $x\in\{0,1\}^N$ to relative accuracy $\eps$
is $\tO\(\frac{1}{\eps^{3/2}}\sqrt{{\cal C}_Z(x)}N^{1/4}\)$.
\end{theorem}

\subsection{NAND-Trees and $st$-Connectivity}\label{sec:nand-graphs}

In this section, we present a useful relationship between the even-depth \textsc{nand}-tree evaluation problem and the $st$-connectivity problem on certain graphs. 
Let $G_0$ be a graph on 2 vertices, labelled $s$ and $t$, with a
single edge. For $d>0$, let $G_d$ be the graph obtained from $G_{d-1}$
by replacing each edge with a 4-cycle --- or more specifically,
replacing an edge $\{u,v\}$ with two paths of length 2 from $u$ to
$v$.  The first three such graphs are illustrated in Figure
\ref{fig:first3}. It is not difficult to see by recursive argument that
the graph $G_d$ has $2^{2d}$ edges.

\begin{figure}[h]
\centering
\begin{tikzpicture}
\node at (0,0) {
\begin{tikzpicture}
\node at (0,1.25) {$s$};
\filldraw (0,1) circle (.1);
\draw (0,0)--(0,1);
\node at (0,-.25) {$t$};
\filldraw (0,0) circle (.1);
\end{tikzpicture}};

\node at (3,0) {
\begin{tikzpicture}
\node at (0,1.75) {$s$};
\filldraw (0,1.5) circle (.1);
\draw (0,1.5)--(-.5,.5);
\draw (0,1.5) -- (.5,.5);
\filldraw (-.5,.5) circle (.1);
\filldraw (.5,.5) circle (.1);
\draw (0,-.5)--(-.5,.5);
\draw (0,-.5) -- (.5,.5);
\filldraw (0,-.5) circle (.1);
\node at (0,-.75) {$t$};
\end{tikzpicture}};

\node at (7,0) {
\begin{tikzpicture}
\node at (0,1.75) {$s$};
\filldraw (0,1.5) circle (.1);

\draw (0,1.5)--(-.75,1);
\draw (0,1.5)--(-.4,.65);
\draw (0,1.5)--(.75,1);
\draw (0,1.5)--(.4,.65);

\filldraw (-.75,1) circle (.1);
\filldraw (-.4,.65) circle (.1);

\filldraw (.75,1) circle (.1);
\filldraw (.4,.65) circle (.1);

\draw (-1,0)--(-.75,1);
\draw (-1,0)--(-.4,.65);
\draw (1,0)--(.75,1);
\draw (1,0)--(.4,.65);

\filldraw (-1,0) circle (.1);
\filldraw (1,0) circle (.1);

\draw (0,-1.5)--(-.75,-1);
\draw (0,-1.5)--(-.4,-.65);
\draw (0,-1.5)--(.75,-1);
\draw (0,-1.5)--(.4,-.65);

\filldraw (-.75,-1) circle (.1);
\filldraw (-.4,-.65) circle (.1);

\filldraw (.75,-1) circle (.1);
\filldraw (.4,-.65) circle (.1);

\draw (-1,0)--(-.75,-1);
\draw (-1,0)--(-.4,-.65);
\draw (1,0)--(.75,-1);
\draw (1,0)--(.4,-.65);

\filldraw (0,-1.5) circle (.1);
\node at (0,-1.75) {$t$};
\end{tikzpicture}};

\node at (0,-1.25) {$G_0$};
\node at (3,-1.75) {$G_1$};
\node at (7,-2.25) {$G_2$};
\end{tikzpicture}
\caption{The graphs $G_0$, $G_1$ and $G_2$.}\label{fig:first3}
\end{figure}
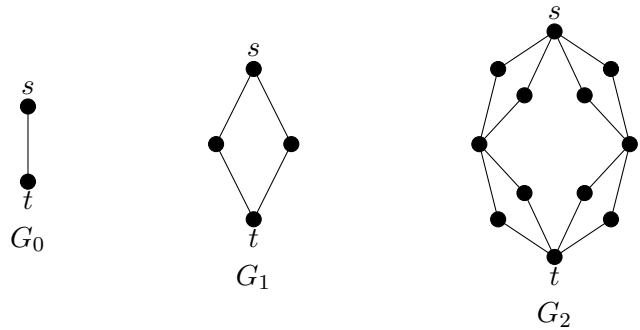

\begin{figure}[h]
\centering
\begin{tikzpicture}
\node at (3,0) {
\begin{tikzpicture}
\node at (0,1.25) {$u$};
\filldraw (0,1) circle (.1);
\draw (0,0)--(0,1);
\node at (.25,.5) {$\tau$};
\filldraw (0,0) circle (.1);
\node at (0,-.25) {$v$};
\end{tikzpicture}
};

\node at (5,0) {\Large$\longmapsto$};
\node at (8,0) {
\begin{tikzpicture}
\node at (0,1.75) {$u$};
\filldraw (0,1.5) circle (.1);
\draw (0,1.5)--(-.5,.5); 
\node at (-1,1) {$\tau 00$};
\draw (0,1.5) -- (.5,.5); 
\node at (-1,0) {$\tau 01$};
\filldraw (-.5,.5) circle (.1);
\filldraw (.5,.5) circle (.1);
\draw (0,-.5)--(-.5,.5); 
\node at (1,1) {$\tau 10$};
\draw (0,-.5) -- (.5,.5); 
\node at (1,0) {$\tau 11$};
\filldraw (0,-.5) circle (.1);
\node at (0,-.75) {$v$};
\end{tikzpicture}
};
\end{tikzpicture}
\caption{This graph shows how to label edges of $G_{d+1}$ given
a label $\tau$ on an edge $\{u,v\}\in E(G_d)$. For example, if $\tau=01$,
the new labels will be $0100,0101,0110,$ and $0111$.}\label{fig:string_label}
\end{figure}
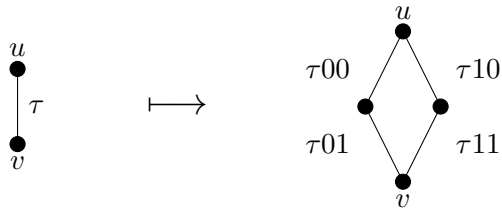

\begin{figure}[h]
\centering
\begin{tikzpicture}
\node at (0,0) {
\begin{tikzpicture}
\node at (0,1.75) {$s$};
\filldraw (0,1.5) circle (.1);
\draw (0,1.5)--(-.5,.5);
\draw (0,1.5) -- (.5,.5);
\filldraw (-.5,.5) circle (.1);
\filldraw (.5,.5) circle (.1);
\draw (0,-.5)--(-.5,.5);
\draw (0,-.5) -- (.5,.5);
\filldraw (0,-.5) circle (.1);
\node at (0,-.75) {$t$};

\node at (-.6,1) {$x_0$};
\node at (.6,1) {$x_2$};
\node at (-.6,0) {$x_1$};
\node at (.6,0) {$x_3$};
\end{tikzpicture}
};

\node at (4,0) {
\begin{tikzpicture}
\filldraw (0,0) circle (.05);
\filldraw (-.5,-1) circle (.05);
\filldraw (.5,-1) circle (.05);
\filldraw (-.75,-2) circle (.05);
\filldraw (-.25,-2) circle (.05);
\filldraw (.75,-2) circle (.05);
\filldraw (.25,-2) circle (.05);

\draw (0,0)--(-.5,-1);
\draw (0,0)--(.5,-1);
\draw (-.5,-1)--(-.75,-2);
\draw (-.5,-1)--(-.25,-2);
\draw (.5,-1)--(.75,-2);
\draw (.5,-1)--(.25,-2);

\node at (-.75,-2.3) {$x_0$};
\node at (-.25,-2.3) {$x_1$};
\node at (.25,-2.3) {$x_2$};
\node at (.75,-2.3) {$x_3$};

\node at (-.25,0) {$\vee$};
\node at (-.75,-1) {$\wedge$};
\node at (.75,-1) {$\wedge$};
\end{tikzpicture}
};

\end{tikzpicture}
\caption{Here we show the relationship between the labeling 
of edges in the graph $G_2$, and the inputs $x_i$ to
 a depth-$2$ \textsc{nand}-tree.}\label{fig:graph-and-tree}
\end{figure}
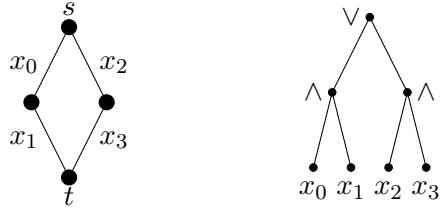

Let $x\in\{0,1\}^{2^{2d}}$ be a depth-$2d$ \textsc{nand}-tree instance.  We can associate the $2^{2d}$ binary
variables $x_i$ with the edges of $G_d$ inductively as
follows. Begin with $d=0$. $G_0$ has one edge,  which we label
with the empty string $\sigma=\emptyset.$ We then associate the variable
$x_\sigma$ with the edge labeled by the string $\sigma$. In this case,
there is only one variable $x_\emptyset$, and it is associated with the edge
labeled by $\emptyset.$

For the inductive step, we have a graph $G_d$, for $d\geq 0$, with
edges labeled by strings. We now want to label the strings of the
graph $G_{d+1}$. To create $G_{d+1}$ from $G_d$ each edge of $G_d$ is
replaced by the graph $G_1$. Consider an edge between vertices 
$u$ and $v$ in the graph $G_d$ that is
labeled by the string $\tau$. When the edge between $u$ and $v$
is replaced by a four-cycle to create $G_{d+1}$, we label the four edges
as in Figure \ref{fig:string_label}.

We can now define a subgraph $G_d(x)$ of $G_d$ by including only those edges of
$G_d$ in which the associated input variable is true. 
This allows us to directly connect instances of depth-$2d$ \textsc{nand}-trees
to subgraphs of $G_d$. In Figure \ref{fig:graph-and-tree}, we show
explicitly how the labels of a graph correspond to the inputs to a \textsc{nand}-tree
for $d=2.$ Then we have the
following:

\begin{lemma}\label{lem:nand-graph}
For all $x\in\{0,1\}^{2^{2d}}$, $x$ is a 1-instance of depth-$2d$ \textsc{nand}-tree evaluation if and only
if $s$ and $t$ are connected in $G_d(x)$.
\end{lemma}
\begin{proof}
We proceed by induction. We start with the depth-0 \textsc{nand}-tree. 
The depth-0 \textsc{nand}-tree is simply the identity function. The graph $G_0(x)$
associated with this tree consists of only the vertices $s$ and $t$: if
$x=1$, there is an edge between $s$ and $t$, and if $x=0$, there
is no edge between $s$ and $t$.
Thus $s$ and $t$ are connected in $G_0(x)$ if and only if the depth-0 \textsc{nand}-tree evaluates to 1
on~$x$.

\begin{figure}[h]
\centering
\begin{tikzpicture}
\tikzstyle{operator} = [rectangle,rounded corners,draw,fill=white,minimum size=1.5em]
\tikzstyle{vertex}=[circle,fill=black,draw,scale=.5]
\node at (1.5,1.5) {$s$};
\node[vertex] (s) at (1.5,1) {};
\node[operator] (n1) at (0,0) {$G_d(x^{00})$};
\node[operator] (n2) at (3,0) {$G_d(x^{10})$};
\node[operator] (n3) at (0,-2) {$G_d(x^{01})$};
\node[operator] (n4) at (3,-2) {$G_d(x^{11})$};
\node[vertex] (t) at (1.5,-3) {};
\node at (1.5,-3.5) {$t$};
\node[vertex] (left) at (-1.5,-1) {};
\node at (-2,-1) {$l$};
\node[vertex] (right) at (4.5,-1) {};
\node at (5,-1) {$r$};
\path (s) edge (n1);
\path (left) edge (n1);
\path (s) edge (n2);
\path (right) edge (n2);
\path (t) edge (n3);
\path (left) edge (n3);
\path (t) edge (n4);
\path (right) edge (n4);
\end{tikzpicture}
\caption{In the above graph we identify $l$ with the original node $t$ of $G_d(x^{00})$,
$l$ with the original node $s$ of $G_d(x^{01})$,
$r$ with the original node $t$ of $G_d(x^{10})$, and
$r$ with the original node $s$ of $G_d(x^{11})$.}
\label{fig:subgraphs}
\end{figure}
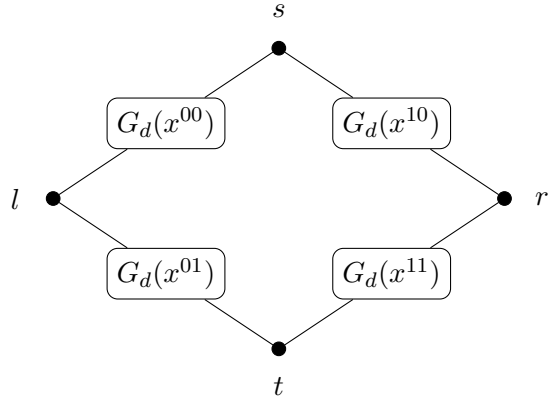

For the induction step, we suppose that the hypothesis is true for 
depth-$2d$ \textsc{nand}-trees, and we would like to prove it is true for
depth-$2(d+1)$ \textsc{nand}-trees. Now the function $\textsc{nand}_{2(d+1)}$ evaluates to true if and only
if 
\begin{align}
\left(\textsc{nand}_{2d}\left(x^{00}\right)=1\wedge
\textsc{nand}_{2d}\left(x^{01}\right)=1\right)\bigvee
\left(\textsc{nand}_{2d}\left(x^{10}\right)=1\wedge
\textsc{nand}_{2d}\left(x^{11}\right)=1\right)=1,
\end{align}
where $x^{\tau}$
indicates all bits of $x$ whose labeling string begins with $\tau$.
But if any of these 
$\textsc{nand}_{2d}\left(x^\tau\right)=1$, by the induction assumption, there must be a
connective path between the nodes at either end of the subgraph
$G_d(x^{\tau})$ in Figure \ref{fig:subgraphs}. This implies that $s$ and
$t$ are connected in $G_{d+1}(x)$ if and only if 
$\textsc{nand}_{2(d+1)}(x)=1$. 
\end{proof}

\paragraph{Dual Graphs} We now define a class of graphs that relate \emph{$0$-instances} of \textsc{nand}-trees to 
$st$-connectivity problems. Define $G_0'$ as an edge with the endpoints labelled $s'$ and $t'$. For $d\geq 1$, define $G_d'$ recursively by replacing every edge in $G_{d-1}'$ with the multigraph consisting of two 2-paths with the same three vertices. The first three such graphs are shown in Figure \ref{fig:dual-nand-graphs}. 

\begin{figure}[h]
\centering
\begin{tikzpicture}

\node at (0,0) {\begin{tikzpicture}
\filldraw (0,0) circle (.1);
\filldraw (1,0) circle (.1);
\draw (0,0) -- (1,0);
\node at (-.25,0) {$s'$};
\node at (1.25,0) {$t'$};
\end{tikzpicture}
};

\node at (3,0) {
\begin{tikzpicture}
\node at (-.25,0) {$s'$};
\filldraw (0,0) circle (.1);
\filldraw (1,0) circle (.1);
\filldraw (2,0) circle (.1);
\node at (2.25,0) {$t'$};

\draw plot [smooth] coordinates {(0,0)  (.5,.25)  (1,0)};
\draw plot [smooth] coordinates {(0,0)  (.5,-.25)  (1,0)};
\draw plot [smooth] coordinates {(2,0)  (1.5,.25)  (1,0)};
\draw plot [smooth] coordinates {(2,0)  (1.5,-.25)  (1,0)};
\end{tikzpicture}
};

\node at (7,0) {
\begin{tikzpicture}
\filldraw (0,0) circle (.1);
\filldraw (1.5,0) circle (.1);
\filldraw (3,0) circle (.1);
\filldraw (.75,.75) circle (.1);
\filldraw (.75,-.75) circle (.1);
\filldraw (2.25,.75) circle (.1);
\filldraw (2.25,-.75) circle (.1);

\draw plot [smooth] coordinates {(0,0) (.25,.5) (.75,.75)};
\draw plot [smooth] coordinates {(3,0) (2.75,.5) (2.25,.75)};
\draw plot [smooth] coordinates {(0,0) (.25,-.5) (.75,-.75)};
\draw plot [smooth] coordinates {(3,0) (2.75,-.5) (2.25,-.75)};

\draw plot [smooth] coordinates {(1.5,0) (1.25,.5) (.75,.75)};
\draw plot [smooth] coordinates {(1.5,0) (1.25,-.5) (.75,-.75)};
\draw plot [smooth] coordinates {(1.5,0) (1.75,-.5) (2.25,-.75)};
\draw plot [smooth] coordinates {(1.5,0) (1.75,.5) (2.25,.75)};

\draw plot [smooth] coordinates {(0,0) (.5,.25) (.75,.75)};
\draw plot [smooth] coordinates {(3,0) (2.5,.25) (2.25,.75)};
\draw plot [smooth] coordinates {(0,0) (.5,-.25) (.75,-.75)};
\draw plot [smooth] coordinates {(3,0) (2.5,-.25) (2.25,-.75)};

\draw plot [smooth] coordinates {(1.5,0) (1,.25) (.75,.75)};
\draw plot [smooth] coordinates {(1.5,0) (1,-.25) (.75,-.75)};
\draw plot [smooth] coordinates {(1.5,0) (2,-.25) (2.25,-.75)};
\draw plot [smooth] coordinates {(1.5,0) (2,.25) (2.25,.75)};

\node at (-.25,0) {$s'$};
\node at (3.25,0) {$t'$};
\end{tikzpicture}
};

\node at (0,-.5) {$G_0'$};
\node at (3,-.75) {$G_1'$};
\node at (7,-1.25) {$G_2'$};

\end{tikzpicture}
\caption{The graphs $G'_0$, $G'_1$ and $G'_2$.}\label{fig:dual-nand-graphs}
\end{figure}
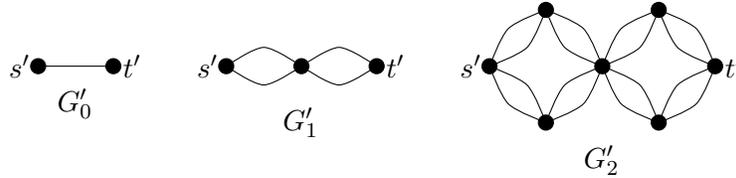

Note that for any $d$, both $G_d$ and $G_d'$ are planar. 
Define $\bar{G}_d$ as a planar embedding of $G_d$ with an additional edge $\{s,t\}$. Define $\bar{G}_d'$ as a planar embedding of $G_d'$ with an additional edge $\{s',t'\}$. Then we have the following.

\begin{theorem}
$\bar{G}_d'=\bar{G}_d^\dagger$, the dual graph of $\bar{G}_d$.
\end{theorem}
\begin{proof}
We will prove the theorem by induction on $d$. For $d=0$, both $\bar{G}_0$ and $\bar{G}_0'$ are just two edges with the same endpoints, and this graph is self-dual.

Note that for all $d\geq 1$, $\bar{G}_d$ has an edge $\{s,t\}$, and $\bar{G}_d'$ has an edge $\{s',t'\}$. These edges will always be dual. 
Suppose $\bar{G}_{d-1}'=\bar{G}_{d-1}^\dagger$. Fix some edge $e=\{u,v\}\in E(\bar{G}_{d-1})$ other than $\{s,t\}$. To get to $\bar{G}_d$, this edge is replaced by 4 edges, $e_0,e_1,e_2,e_3$ as in Figure \ref{fig:dual-proof} (a).

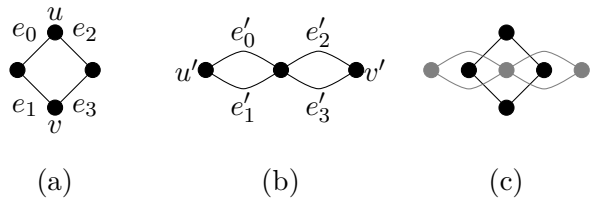
\begin{figure}
\centering
\begin{tikzpicture}
\node at (0,0) {
\begin{tikzpicture}
\filldraw (0,0) circle (.1);
\filldraw (.5,.5) circle (.1);
\filldraw (.5,-.5) circle (.1);
\filldraw (1,0) circle (.1);
\draw (0,0)--(.5,.5)--(1,0)--(.5,-.5)--(0,0);

\node at (.1,.5) {$e_0$};
\node at (.1,-.5) {$e_1$};
\node at (.9,.5) {$e_2$};
\node at (.9,-.5) {$e_3$};

\node at (.5,.75) {$u$};
\node at (.5,-.75) {$v$};
\end{tikzpicture}
};

\node at (3,0) {
\begin{tikzpicture}
\node at (-.25,0) {$u'$};
\filldraw (0,0) circle (.1);
\filldraw (1,0) circle (.1);
\filldraw (2,0) circle (.1);
\node at (2.25,0) {$v'$};

\draw plot [smooth] coordinates {(0,0)  (.5,.25)  (1,0)};
\draw plot [smooth] coordinates {(0,0)  (.5,-.25)  (1,0)};
\draw plot [smooth] coordinates {(2,0)  (1.5,.25)  (1,0)};
\draw plot [smooth] coordinates {(2,0)  (1.5,-.25)  (1,0)};

\node at (.5,.5) {$e_0'$};
\node at (.5,-.5) {$e_1'$};
\node at (1.5,.5) {$e_2'$};
\node at (1.5,-.5) {$e_3'$};
\end{tikzpicture}
};

\node at (6,0) {
\begin{tikzpicture}
\filldraw (.5,0) circle (.1);
\filldraw (1,.5) circle (.1);
\filldraw (1,-.5) circle (.1);
\filldraw (1.5,0) circle (.1);
\draw (.5,0)--(1,.5)--(1.5,0)--(1,-.5)--(.5,0);

\filldraw[gray] (0,0) circle (.1);
\filldraw[gray] (1,0) circle (.1);
\filldraw[gray] (2,0) circle (.1);

\draw[gray] plot [smooth] coordinates {(0,0)  (.5,.25)  (1,0)};
\draw[gray] plot [smooth] coordinates {(0,0)  (.5,-.25)  (1,0)};
\draw[gray] plot [smooth] coordinates {(2,0)  (1.5,.25)  (1,0)};
\draw[gray] plot [smooth] coordinates {(2,0)  (1.5,-.25)  (1,0)};
\end{tikzpicture}
};

\node at (0,-1.5) {(a)};
\node at (3,-1.5) {(b)};
\node at (6,-1.5) {(c)};
\end{tikzpicture}
\caption{(a) shows a subgraph of $G_d$ corresponding to an edge $e=\{u,v\}$ of $G_{d-1}$. (b) shows the subgraph of $G_d'$ corresponding to the dual edge of $e$, $e^\dagger = \{u',v'\}$, in $G_{d-1}'$, $(u',v')$. (c) shows how the two subgraphs are dual.}\label{fig:dual-proof}
\end{figure}

By the induction hypothesis, $e$ has a dual edge $e^\dagger$ in $E(\bar{G}_{d-1}')$. To obtain $\bar{G}_{d}'$ from $\bar{G}_{d-1}'$, we replace $e^\dagger=\{u',v'\}$ with 4 edges, $e_0',e_1',e_2',e_3'$ as in Figure \ref{fig:dual-proof} (b).
Replacing $e$ has introduced two new vertices, and one new face, whereas replacing $e^\dagger$ has introduced two new faces and one new vertex. We identify them as in Figure \ref{fig:dual-proof} (c).
We can thus see that $e_b'=e_b^\dagger$ for all $b\in\{0,1,2,3\}$.
\end{proof}
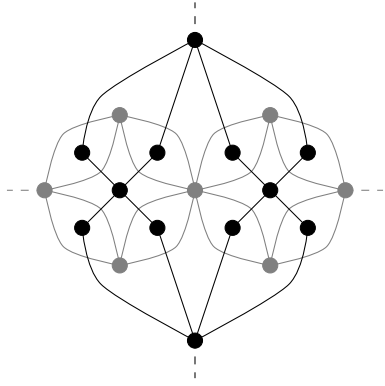
\begin{figure}[h!]
\centering
\begin{tikzpicture}

\filldraw (0,2) circle (.1);

\filldraw[gray] (-1,1) circle (.1);

\filldraw (-1.5,.5) circle (.1);
\filldraw (-.5,.5) circle (.1);

\filldraw[gray] (-2,0) circle (.1);
\filldraw (-1,0) circle (.1);
\filldraw[gray] (0,0) circle (.1);

\filldraw[gray] (-1,-1) circle (.1);

\filldraw (-1.5,-.5) circle (.1);
\filldraw (-.5,-.5) circle (.1);

\filldraw (0,-2) circle (.1);

\filldraw[gray] (1,1) circle (.1);

\filldraw (1.5,.5) circle (.1);
\filldraw (.5,.5) circle (.1);

\filldraw[gray] (2,0) circle (.1);
\filldraw (1,0) circle (.1);

\filldraw[gray] (1,-1) circle (.1);

\filldraw (1.5,-.5) circle (.1);
\filldraw (.5,-.5) circle (.1);

\draw[gray] plot [smooth] coordinates {(-2,0) (-1.75,.75) (-1,1)};
\draw[gray] plot [smooth] coordinates {(2,0) (1.75,.75) (1,1)};
\draw[gray] plot [smooth] coordinates {(-2,0) (-1.75,-.75) (-1,-1)};
\draw[gray] plot [smooth] coordinates {(2,0) (1.75,-.75) (1,-1)};

\draw[gray] plot [smooth] coordinates {(2,0) (1.25,-.25) (1,-1)};
\draw[gray] plot [smooth] coordinates {(2,0) (1.25,.25) (1,1)};
\draw[gray] plot [smooth] coordinates {(-2,0) (-1.25,-.25) (-1,-1)};
\draw[gray] plot [smooth] coordinates {(-2,0) (-1.25,.25) (-1,1)};

\draw[gray] plot [smooth] coordinates {(0,0) (-.25,.75) (-1,1)};
\draw[gray] plot [smooth] coordinates {(0,0) (.25,.75) (1,1)};
\draw[gray] plot [smooth] coordinates {(0,0) (-.25,-.75) (-1,-1)};
\draw[gray] plot [smooth] coordinates {(0,0) (.25,-.75) (1,-1)};

\draw[gray] plot [smooth] coordinates {(0,0) (-.75,.25) (-1,1)};
\draw[gray] plot [smooth] coordinates {(0,0) (.75,.25) (1,1)};
\draw[gray] plot [smooth] coordinates {(0,0) (-.75,-.25) (-1,-1)};
\draw[gray] plot [smooth] coordinates {(0,0) (.75,-.25) (1,-1)};

\draw (-1.5,.5)--(-1,0);
\draw (-1,0) -- (-.5,.5);
\draw (-1,0) -- (-.5,-.5);
\draw (-1,0) -- (-1.5,-.5);

\draw (1.5,.5)--(1,0);
\draw (1,0) -- (.5,.5);
\draw (1,0) -- (.5,-.5);
\draw (1,0) -- (1.5,-.5);

\draw (0,2)--(.5,.5);
\draw (0,2) -- (-.5,.5);
\draw (0,-2) -- (-.5,-.5);
\draw (0,-2) -- (.5,-.5);

\draw plot [smooth] coordinates {(0,2) (-1.25,1.25) (-1.5,.5)}; 
\draw plot [smooth] coordinates {(0,-2) (-1.25,-1.25) (-1.5,-.5)}; 
\draw plot [smooth] coordinates {(0,2) (1.25,1.25) (1.5,.5)}; 
\draw plot [smooth] coordinates {(0,-2) (1.25,-1.25) (1.5,-.5)}; 

\draw[dashed] (0,2)--(0,2.5);
\draw[dashed] (0,-2)--(0,-2.5);

\draw[gray,dashed] (-2,0) -- (-2.5,0);
\draw[gray,dashed] (2,0) -- (2.5,0);
\end{tikzpicture}
\caption{$\bar{G}_2$ and its dual $\bar{G}_2'$.}\label{fig:G2_dual}
\end{figure}
\noindent In Figure \ref{fig:G2_dual}, we show how $\bar{G}_d$ and $\bar{G}'_d$ are dual for
for $d=2.$

Recall that each edge $e\in E(G_d)$ corresponds to some input variable
$x_e$. Thus every edge $e\in E(G_d')$ corresponds to an input variable
$x_{e^\dagger}$. For any instance $x\in\{0,1\}^{2^{2d}}$, we define a
subgraph of $G_d'$, $G_d'(x)$, by including only those edges of $G_d'$
for which the corresponding variable $x_{e^\dagger}$ is 0. Just as
$\textsc{nand}_{2d}(x)=1$ if and only if $s$ and $t$ are connected in
$G_d(x)$, we now show $\textsc{nand}_{2d}({x})=0$ if and only if $s'$ and $t'$ are
connected in $G'_d(x)$.

\begin{lemma}
For any $x\in\{0,1\}^{2^{2d}}$, $\textsc{nand}_{2d}({x})=0$ if and only if $s'$ and $t'$ are connected in $G'_d(x)$.
\end{lemma}
\begin{proof}
The proof is very similar to the proof of Lemma \ref{lem:nand-graph}, so we merely sketch the argument. Define $\overline{\textsc{nand}}_{2d}$ as the formula on a full binary tree of depth $2d$ in which we label vertices at even distance from the leaves by $\wedge$, and those at odd distance from the leaves by $\vee$ (the opposite of a \textsc{nand}-tree). 
First, we have by De Morgan's Law, $\textsc{nand}_{2d}({x})=0$ if and only if $\overline{\textsc{nand}}_{2d}(\bar{x})=1$, where $\bar{x}$ is the entrywise complement of $x$. 
A simple inductive proof shows that $\bar{x}$ is a 1-instance of $\overline{\textsc{nand}}_{2d}$ if and only if $s'$ and $t'$ are connected in $G_d'(x)$. 
\end{proof}

\paragraph{Connection to Average Choice Complexity}
We have described a connection between evaluating \textsc{nand}-trees and solving $st$-connectivity problems. We will now connect the effective resistance of these graphs to average choice complexity, defined in Section \ref{sec:average-choice-complexity}. 
\begin{lemma}\label{lemma:kfault_conn}
For any instance of \textsc{nand}-tree, $x\in\{0,1\}^{2^{2d}}$, ${\cal C}_A(x)=R_{s,t}(G_d(x))$ and ${\cal C}_B(x)=R_{s',t'}(G_d'(x))$. Furthermore, for all $x$, ${\cal C}(x)\leq {\cal F}(x)$, where ${\cal F}$ is the fault complexity. 
\end{lemma}
\begin{proof}
We will give a proof for ${\cal C}_A$, as the case of ${\cal C}_B$ is similar. 

First, $R_{s,t}(G_d(x))=\infty$ if and only if $s$ and $t$ are not connected in $G_d(x)$, if and only if $x$ is a 0-instance, if and only if ${\cal C}_A(x)=\infty$. Thus, suppose this is not the case, so ${\cal C}(x)={\cal C}_A(x)<\infty$.

The rest of the proof is by induction. For the case of $d=0$, we have to consider the only $A$-winnable input in $\{0,1\}^{2^0}$, $x=1$. In that case, we have ${\cal C}_A(x)=1$, and since $G_0(x)$ is just a single edge from $s$ to $t$, $R_{s,t}(G_0(x))=1$. We also have ${\cal F}(x)=1$, since there are no choices, so there are no faults.

Let $x\in\{0,1\}^{2^{2(d+1)}}$ be any $A$-winnable input. Let
$\mathscr{S}_A(x)$ be the set of all $A$-winning strategies on input $x$, 
and let $\cal S$ be a strategy in $\mathscr{S}_A(x)$. Let $b\in\{0,1\}$
be the choice of $\cal S$ at the root. Then
$\cal S$ can be
described recursively by $b$ and a pair
of winning sub-strategies, ${\cal S}^0\in \mathscr{S}_A({x^{b0}})$ and ${\cal
S}^1\in\mathscr{S}_A({x^{b1}})$, one for each of the possible first
choices of Player $B$.  If $r$ is the root, we have:
\begin{eqnarray}
{\cal C}_A(x) &=& \min_{{\cal S}\in\mathscr{S}_A(x)}\mathbb{E}_{p\in {\cal S}}\left[\prod_{v\in\nu_A(p)}c(v)\right]\notag\\
&=& \min_{b\in\{0,1\}}\min_{{\cal S}^0\in\mathscr{S}_A(x^{b0}),{\cal S}^1\in\mathscr{S}_A(x^{b1})}\frac{1}{2}\(\mathbb{E}_{p\in {\cal S}^0}\left[c(r)\prod_{v\in\nu_A(p)}c(v)\right]
+\mathbb{E}_{p\in {\cal S}^1}\left[c(r)\prod_{v\in\nu_A(p)}c(v)\right]\)\notag\\
&=& \frac{1}{2}c(r)\min_{b\in\{0,1\}}\{{\cal C}_A(x^{b0})+{\cal C}_A(x^{b1})\}.\label{eq:CR}
\end{eqnarray}
We first consider the case that $r$ is a fault, so $c(r)=2$. We first show ${\cal C}_A(x)\leq {\cal F}(x)$.
Using the induction hypothesis,
\begin{align}
{\cal C}_A(x)\leq \min_{b\in\{0,1\}}\{{\cal F}_A(x^{b0})+{\cal F}_A(x^{b1})\}
\leq \min_{b}\max_{b'}2{\cal F}_A(x^{bb'})= {\cal F}(x).
\end{align}
Also by the induction hypothesis:
\begin{equation}
{\cal C}_A(x)=\min_{b\in\{0,1\}}\{R_{s,t}(G_{d}(x^{b0}))+R_{s,t}(G_{d}(x^{b1}))\}.
\end{equation}
Since $r$ is a fault, exactly one of ${\cal C}_A(x^{b0})+{\cal
C}_A(x^{b1})$ for $b\in\{0,1\}$ is finite. Without loss of generality,
suppose it is $b=0$. Then
$R_{s,t}(G_{d}(x^{10}))+R_{s,t}(G_{d}(x^{11}))=\infty$, meaning
that $s$ and $t$ are not connected in one of $G_{d}(x^{10})$ and
$G_{d}(x^{11})$. Referring to Figure \ref{fig:subgraphs}, which shows
how $G_{d+1}(x)$ is be constructed from $G_d(x^{00})$, $G_d(x^{01})$,
$G_d(x^{10})$ and $G_d(x^{11})$, since resistances in
series add, we have $R_{s,t}(G_{d+1}(x))=R_{s,t}(G_{d}(x^{00}))+R_{s,t
}(G_{d}(x^{01}))={\cal C}_A(x)$, as desired.

Now suppose that $r$ is not a fault. Continuing from \eqref{eq:CR}, we have:
\begin{eqnarray}\label{eq:C_nofault}
{\cal C}_A(x)&=& \frac{1}{2}\frac{\max\{{\cal C}(x^{00})+{\cal C}(x^{01}),{\cal C}(x^{10})+{\cal C}(x^{11})\}}{\mathrm{avg}\{{\cal C}(x^{00})+{\cal C}(x^{01}),{\cal C}(x^{10})+{\cal C}(x^{11})\}}\min_{b\in \{0,1\}}\left\{{\cal C}(x^{b0})+{\cal C}(x^{b1})\right\}\notag\\
&=& \frac{({\cal C}(x^{00})+{\cal C}(x^{01}))({\cal C}(x^{10})+{\cal C}(x^{11}))}{{\cal C}(x^{00})+{\cal C}(x^{01})+{\cal C}(x^{10})+{\cal C}(x^{11})}.\label{eq:C}
\end{eqnarray}
Since $v$ is not a fault, we have ${\cal F}(x)=\max_{b,b'\in\{0,1\}}{\cal F}(x^{bb'})$. Using \eqref{eq:C_nofault}
and the induction hypothesis, we get:
\begin{equation}
{\cal C}_A(x)\leq \frac{1}{2}\max_{b\in\{0,1\}}\{{\cal C}_A(x^{b0})+{\cal C}_A(x^{b1})\}\leq \frac{1}{2}\max_{b\in\{0,1\}}\{{\cal F}(x^{b0})+{\cal F}(x^{b1})\}\leq {\cal F}(x).
\end{equation}
By the induction hypothesis and \eqref{eq:C}, we get:
\begin{equation}
{\cal C}_A(x) = \frac{(R_{s,t}(G_d(x^{00}))+R_{s,t}(G_d(x^{01})))(R_{s,t}(G_d(x^{10}))+R_{s,t}(G_d(x^{11})))}{R_{s,t}(G_d(x^{00}))+R_{s,t}(G_d(x^{01}))+R_{s,t}(G_d(x^{10}))+R_{s,t}(G_d(x^{11}))}.
\end{equation}
To complete the proof, we refer to Figure \ref{fig:subgraphs}, which shows how $G_{d+1}(x)$ can be constructed from $G_d(x^{00})$, $G_d(x^{01})$, $G_d(x^{10})$ and $G_d(x^{11})$. Since resistances in series add, and inverse resistances in parallel add, we have:
\begin{eqnarray}
\frac{1}{R_{s,t}(G_{d+1}(x))} &=& \frac{1}{R_{s,t}(G_d(x^{00}))+R_{s,t}(G_d(x^{01}))} +\frac{1}{R_{s,t}(G_d(x^{01}))+R_{s,t}(G_d(x^{11}))}\nonumber\\
&=& \frac{R_{s,t}(G_d(x^{00}))+R_{s,t}(G_d(x^{01}))+R_{s,t}(G_d(x^{00}))+R_{s,t}(G_d(x^{01}))}{{R_{s,t}(G_d(x^{00}))+R_{s,t}(G_d(x^{01}))})({R_{s,t}(G_d(x^{00}))+R_{s,t}(G_d(x^{01}))})},
\end{eqnarray}
so ${\cal C}_A(x)=R_{s,t}(G_{d+1}(x))$. 

A similar analysis for ${\cal C}_B$ completes the proof. 
\end{proof}


\subsection{Span Program for NAND-Trees}\label{sec:span}

We use the relationship between $st$-connectivity on $G_d$ and depth-$2d$ \textsc{nand}-tree evaluation
to define a span program for \textsc{nand}-tree evaluation.
It is nearly identical to a span program for $st$-connectivity given
in \cite{BR12} (see also \cite[Section 4]{IJ15}), except that we
exploit the fact that we are only  considering subgraphs of $G_d$, so
not all edges are possible. We call the following span program
$P_{G_d}$:
\begin{equation*}
H_{e,0}=\{0\},\quad H_{e,1}=\mathrm{span}\{\ket{e}\},\quad H=\mathrm{span}\{\ket{e}:e\in \overrightarrow{E}(G_d)\},\quad V= \mathrm{span}\{\ket{u}:u\in V(G_d)\}
\end{equation*}
\begin{equation*}
\tau=\ket{s}-\ket{t}, \quad A=\sum_{(u,v)\in \overrightarrow{E}(G_d)}(\ket{u}-\ket{v})\bra{u,v}.
\end{equation*}
We have used $e$ in $H_{e,1}$ as a short-hand for the input variable associated with $e$.

\begin{lemma}\label{lemma:pos_witness}
Let $x\in\{0,1\}^{E(G_d)}$ be a 1-instance of \textsc{nand}-tree evaluation, and let $G_d(x)$ be the associated subgraph of $G_d$.
Then $w_+(x,P_{G_d})=\frac{1}{2}R_{s,t}(G_d(x))$.
\end{lemma}
\begin{proof}
Let ${P}_{K}=(\tilde{H},\tilde{V},\tilde{\tau},\tilde{A})$ be the span program for $st$-connectivity on any $|V(G_d)|$-vertex graph, defined in \cite[Section 4]{IJ15}. Then it is easy to see that $\tau=\tilde\tau$, and $\tilde{A}\Pi_{\tilde{H}(G_d(x))}=A\Pi_{H(x)}$, so $w_+(x,P_{G_d})=w_+(G_d(x),{P}_K)$. By \cite[Lemma 4.1]{IJ15}, $w_+(G_d(x),{P}_K)=\frac{1}{2}R_{s,t}(G_d(x))$. 
\end{proof}

We now show that the negative witness size of the span program for $st$-connectivity on $G_d(x)$ is given
by the effective resistance of the dual graph $G_d'(x):$
\begin{lemma}\label{lemma:negative_witness}
Let $x\in\{0,1\}^{E(G_d)}$ be a  0-instance of \textsc{nand}-tree, and let $G_d'(x)$ be the associated subgraph of $G_d'$. Then $w_-(x, P_{G_d})={2}R_{s',t'}(G_d'(x))$. 
\end{lemma}
\begin{proof}
The proof exploits the duality between an $st$-path and an $st$-cut. 

For $x\in\{0,1\}^{2^{2d}}$, define $\bar{G}_d(x)$ (respectively $\bar{G}_d'(x)$) as a planar embedding of $G_d(x)$ (resp.\ $G_d'(x)$) with an additional edge $\{s,t\}$ (resp.\ $\{s',t'\}$). 
We first show that a negative witness for $x$ in $P_{G_d}$ corresponds
to a unit $s't'$-flow in $G_d'(x)$ whose energy equals twice the negative
witness size. Fix a negative witness $\omega$ for $x$, and define a
function $\theta:\overrightarrow{E}(\bar{G}_d')\rightarrow\mathbb{R}$ by
$\theta((u,v)^\dagger)=\omega(u)-\omega(v)$. Then clearly we have
$\theta(u',v')=-\theta(v',u')$ for all $\{u',v'\}\in E(\bar{G}_d')$.
Since $\omega$ is a negative witness for $x$, $\sum_{\{u,v\}\in
E(G_d(x))}(\omega(u)-\omega(v))^2=0$, so whenever $\{u,v\}\in
E({G}_d(x))$, $\omega(u)=\omega(v)$, and thus
$\theta((u,v)^\dagger)=0$. This happens precisely when
$\{u,v\}^\dagger\not\in E(G_d'(x))$. Thus
$\theta$ is only nonzero on the edges of $\bar{G}_d'(x)$. 

\begin{figure}[h]
\centering
\begin{tikzpicture}[scale=1]
\filldraw[gray] (0,0) circle (.1);

\draw[->] (-1,0)--(-.5,.866);
\draw[->] (-.5,.866)--(.5,.866);
\draw[->] (.5,.866)--(1,0);
\draw[->] (1,0)--(.5,-.866);
\draw[->] (.5,-.866)--(-.5,-.866);
\draw[->] (-.5,-.866)--(-1,0);

\filldraw (-1,0) circle (.1);
\filldraw (-.5,.866) circle (.1);
\filldraw (.5,.866) circle (.1);
\filldraw (1,0) circle (.1);
\filldraw (.5,-.866) circle (.1);
\filldraw (-.5,-.866) circle (.1);

\draw[gray,->] (0,0) -- (-1.299,.75);
\draw[gray,->] (0,0) -- (0,1.5);
\draw[gray,->] (0,0) -- (1.299,.75);
\draw[gray,->] (0,0) -- (1.299,-.75);
\draw[gray,->] (0,0) -- (0,-1.5);
\draw[gray,->] (0,0) -- (-1.299,-.75);

\node at (.3,0) {$v'$};
\node at (1.6,.9) {$u'$};

\node at (-1.35,0) {$w_1$};
\node at (-.675,1.1691) {$w_2$};
\node at (.675,1.1691) {$w_3$};
\node at (1.35,0) {$w_4$};
\node at (.675,-1.1691) {$w_5$};
\node at (-.675,-1.1691) {$w_6$};

\end{tikzpicture}
\caption{The dualtiy between a cycle and a star.}\label{fig:face}
\end{figure}
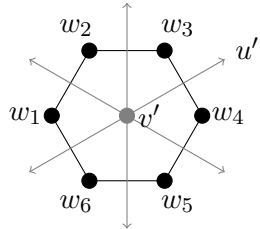

Next, we notice that for any $v'\in
V(\bar{G}_d')$, by the definition of the dual graph, the nodes
$u'$ such that $\{u',v'\}\in \bar{G}_d'$ 
correspond to edges around a face, $f_{v'}$, of
$\bar{G}_d$ (see Figure \ref{fig:face}). If $(w_1,\dots,w_k=w_1)$ 
are the nodes in $\bar{G}_d$ forming the cycle around the face $f_{v'}$, then
\begin{align}
0=\sum_{i=1}^k(\omega(w_i)-\omega(w_{i+1}))=\sum_{i=1}^{k}\theta((w_i,w_{i+1})^
\dagger)=\sum_{u':\{v',u'\}\in E(\bar{G}_d')}\theta(v',u')
=\sum_{u':\{v',u'\}\in E(\bar{G}_d'(x))}\theta(v',u').
\end{align}
Thus, $\theta$ is a circulant on $\bar{G}_d'(x)$, so if we define $\theta'$ so that $\theta'(s',t')=0$, and
$\theta'(v',u')=\theta(v',u')$ when $\{v',u'\}\neq\{s',t'\}$, then $\theta'$
is an $s't'$-flow on $G_d'(x)$. Furthermore, since $\theta(s',t')=\omega(s)-\omega(t)=1$ (since $\omega\tau=1$), $\theta'$ is a \emph{unit} flow. 
The energy of $\theta'$ is:
\begin{equation}
\sum_{\{v',u'\}\in {E}(G_d'(x))}\theta(v',u')^2
=\frac{1}{2}\sum_{(v',u')\in \overrightarrow{E}(G_d')}\theta(v',u')^2=
\frac{1}{2}\sum_{(v',u')\in \overrightarrow{E}(G_d)}(\omega(v')-\omega(u'))^2=\frac{1}{2}\norm{\omega A}^2,
\end{equation}
where the factor of $\frac{1}{2}$ comes from having to sum each edge twice (once for each direction).
Thus $w_-(x,P_{G_d})\geq {2}R_{s',t'}(G_d'(x))$.

Next we show that every unit $s't'$-flow in $G_d'$ corresponds to a
negative witness for $x$ in $P$ whose witness size is equal to half
the energy of the flow. Let $\theta'$ be a unit flow on $G_d'(x)$ from
$s'$ to $t'$. Let $\theta(u,v)=\theta'(u,v)$ when $\{u,v\}\neq
\{s',t'\}$ and $\theta(s',t')=1$. Since $\theta'$ is a unit
$s't'$-flow, $\theta$ is a circulant on $\bar{G}'_d(x)$. It is thus
also a circulant on $\bar{G}'_d$, and so it can be decomposed into a
combination of cyclic-flows around the faces of $\bar{G}'_d$. 

To see this decomposition, let $F$ be the set of faces of $\bar{G}_d'$ with \emph{clockwise}
orientation, and $F'$ is the same faces but with \emph{counter-
clockwise} orientation. Then we can write $\ket{\theta}
=\sum_{e}\theta(e)\ket{e}$ as $\ket{\theta}= \sum_{f\in F\cup
F'}\alpha_f\ket{C_f}$ where $f\in F\cup F'$ and
$\alpha_f\in \mathbb R$, and where
$\ket{C_f}=\sum_{i=1}^{k-1}\ket{w_i,w_{i+1}}$ with $f=(w_1,\dots,w_k=w_1)$ for a
set of vertices $w_i\in\bar{G}_d'$.

For a clockwise oriented face $\overrightarrow{f}\in F$, let
$\overleftarrow{f}$ be the counter-clockwise orientation of the same
face.  Every face $\overrightarrow{f}\in F$ corresponds to a vertex
$v_{{f}}\in V(\bar{G}_d)$, so define 
$\omega(v_f):=\frac{1}{2}(\alpha_{\overrightarrow{f}}-\alpha_{\overleftarrow{f}})$. 
We claim that $\omega$ is a
negative witness. Let $(u',v')$ be any directed edge in
$\overrightarrow{E}(G_d')$. Then since $\{u',v'\}$ is adjacent to two
faces, $(u',v')$ is part of one face in $F$, and one face in $F'$. Let
$\overrightarrow{f}$ be the clockwise oriented face 
containing $(u',v')$, and $\overleftarrow{g}$ be the counter-clockwise oriented
face containing
$(u',v')$. Since these are the only faces containing $(u',v')$, we
must have  $\theta(u',v')=\alpha_{\overrightarrow{f}}+\alpha_{\overleftarrow{g}}$. 
Since $\theta(u',v')=-\theta(v',u')$, we have 
$\alpha_{\overrightarrow{f}}+\alpha_{\overleftarrow{g}}=-\alpha_{\overleftarrow{
f}}-\alpha_{\overrightarrow{g}}$.  Thus:
\begin{equation}
\omega(v_{f})-\omega(v_{g})=\frac{1}{2}(\alpha_{\overrightarrow{f}}-\alpha_{\overleftarrow{f}}-\alpha_{\overrightarrow{g}}
+\alpha_{\overleftarrow{g}})
=\frac{1}{2}(\theta(u',v')-\theta(v',u'))=\theta(u',v').
\end{equation}
So for every edge $(u,v)\in \overrightarrow{E}(\bar{G}_d)$, $\omega(u)-\omega(v)$ is exactly the flow across $(u,v)^\dagger$. Thus for all $\{u,v\}\in E(G_d(x))$, $\{u,v\}^\dagger\not\in E(G_d'(x))$, so $\omega(u)-\omega(v)=0$. Furthermore, $\omega(s)-\omega(t)=1$. Thus $\omega$ is a negative witness, and we have:
\begin{equation}
\norm{\omega A}^2=\sum_{(u,v)\in \overrightarrow{E}(G_d)}(\omega(u)-\omega(v))^2
=2\sum_{e^\dagger \in E(G_d')}\theta(e^\dagger)^2=2J(\theta).
\end{equation}
Because the smallest possible energy of any unit $s't'$-flow is $R_{s',t'}(G'_d(x))$,
we have that $w_-(x)\leq 2R_{s',t'}(G'_d(x))$, completing the proof.
\end{proof}

\noindent Theorem \ref{thm:easy-instances} now follows immediately.
\begin{proof}[{Proof of Theorem \ref{thm:easy-instances}}] Let $f$ be the problem of evaluating \textsc{nand}-trees promised to have ${\cal C}(x)\leq W$.
By Lemmas \ref{lemma:kfault_conn}, \ref{lemma:pos_witness} and \ref{lemma:negative_witness}, we have, for all $1$-instances of $f$, $w_+(x,P_{G_d}) =\frac{1}{2}R_{s,t}(G_d(x))=\frac{1}{2}{\cal C}(x)\leq \frac{1}{2}W$, and for all $0$-instances of $f$ $w_-(x,P_{G_d})=2R_{s',t'}(G'_d(x))=2{\cal C}(x)\leq 2W$. Thus, by Theorem~\ref{thm:span-decision}, the bounded error quantum query complexity of $f$ is at most $O(\sqrt{W_+(f,P_{G_D})W_-(f,P_{G_d})})=O(\sqrt{W^2})=O(W)$.  
\end{proof}

\vskip10pt
In Section \ref{sec:win_nand_tree}, we will consider a different problem: winning the two-player game associated with a \textsc{nand}-tree instance $x$ against an adversary making random choices. To solve this problem, it will be useful not only to evaluate a \textsc{nand}-tree, but to estimate the span program witness size, $w_+(x)$ (or $w_-(x)$). By Theorem \ref{thm:span-est}, we can construct such an algorithm fom $P_{G_d}$. In order to analyze this algorithm's query complexity, we upper bound the approximate negative witness size of inputs to~$P_{G_d}$:
\begin{lemma}\label{lemma:approx_neg_wit}
$\widetilde{W}_-(P_{G_d})\leq 2^{d+1}$ and $\widetilde{W}_+(P_{G_d})\leq 2^{d+1}$. 
\end{lemma}
\begin{proof}For negative inputs, $\tilde{w}_-(x)=w_-(x)=2R_{s',t'}(G'_d(x))\leq 2^{d+1}$, since, it is easy to see by induction, every self-avoiding $s't'$-path in $G_d'$ has length $2^d$.
Thus, we limit ourselves to positive inputs. 
 
Recall that an approximate
negative witness is a function
$\omega:V(G_d)\rightarrow \mathbb{R}$ that minimizes $\|\omega
A\Pi_{H(x)}\|^2=\sum_{(u,v)\in\overrightarrow{E}(G_d(x))}(\omega(u)-\omega(v))^2$ and satisfies $\omega\tau = \omega(s)-\omega(t)=1.$ 
Then a valid
$\omega$ is a voltage induced by a unit
potential difference between $s$ and $t$ in the resistor network
$G_d(x)$ (Dirichlet's principle, or see \cite{DS84}). Without loss of generality, we 
can assign $\omega$ such that $\omega(s)=1$, $\omega(t)=0$, and then on the component of $G_d(x)$ containing $s$ and $t$, $\omega$ is the unique harmonic function with these boundary conditions, and on any other component of $G_d(x)$, $\omega$ is constant.  Then the negative
approximate witness size is
\begin{align}
\tilde{w}_-(x)=\min_{\textrm{valid }\omega}\left\{\|\omega A\|^2=
\sum_{(u,v)\in \overrightarrow{E}(G_d)}(\omega(u)-\omega(v))^2\right\}.
\end{align}

We prove the result inductively. We start with $d=0$. In this case,
the graph consists only of nodes $s$ and $t$, which must take values
$1$ and $0$, so 
$\|\omega A\|^2=(\omega(s)-\omega(t))^2+(\omega(t)-\omega(s))^2=2$.

Now we look at the case of $G_{d+1}(x)$, which
is formed by replacing the edges of a four cycle with graphs
$G_d(x^\tau)$ as in Figure \ref{fig:subgraphs}. For any approximate negative witness $\tilde{\omega}$ for $x$,
let $\gamma(\tilde\omega)=\tilde\omega(l)$ and $\gamma'(\tilde\omega)=\tilde\omega(r)$.
Let $\omega$ be an optimal approximate negative witness for $x$, and we denote $\gamma\equiv \gamma(\omega)$
$\gamma'\equiv \gamma(\omega)$. For $b,b'\in\{0,1\}$, let $\omega_{bb'}$ be the restriction of $\omega$ to $V(G_d(x^{bb'}))$
that takes value $0$ on vertices not in $V(G_d(x^{bb'}))$. 
Then we have:
\begin{eqnarray}
\norm{\omega A}^2 
&=& \sum_{(u,v)\in\overrightarrow{E}(G_d)}(\omega_{00}(u)-\omega_{00}(v))^2+\sum_{(u,v)\in\overrightarrow{E}(G_d)}(\omega_{01}(u)-\omega_{01}(v))^2\notag\\
&&\qquad+\sum_{(u,v)\in \overrightarrow{E}(G_d)}(\omega_{10}(u)-\omega_{10}(v))^2+\sum_{(u,v)\in \overrightarrow{E}(G_d)}(\omega_{11}(u)-\omega_{11}(v))^2.\label{eq:omegaA}
\end{eqnarray}
Consider $\omega_{00}$. Since $\omega$ is a harmonic function on $G_{d+1}(x)$ with boundary conditions $\omega(s)=1$ and $\omega(t)=0$ that minimizes $\norm{\omega A}^2$, $\omega_{00}$ is a harmonic function on $G_d(x^{00})$ with boundary conditions $\omega_{00}(s)=1$ and $\omega_{00}(l)=\gamma$ that minimizes $\sum_{(u,v)\in \overrightarrow{E}(G_d)}(\omega_{00}(u)-\omega_{00}(v))^2$. If $\gamma=1$, $\omega_{00}$ is necessarily constant, and so $\sum_{(u,v)\in \overrightarrow{E}(G_d)}(\omega_{00}(u)-\omega_{00}(v))^2=0$. Otherwise, $\frac{1}{1-\gamma}\omega_{00}$ is a harmonic function on $G_d(x^{00})$ with boundary conditions satisfying $\frac{1}{1-\gamma}\omega_{00}(s)-\frac{1}{1-\gamma}\omega_{00}(l)=1$, minimizing $\sum_{(u,v)\in \overrightarrow{E}(G_d)}(\omega(u)-\omega(v))^2$, so $\frac{1}{1-\gamma}\omega_{00}$ is an optimal approximate negative witness for $x^{00}$, so by the induction hypothesis, we have:
\begin{equation}
\sum_{(u,v)\in \overrightarrow{E}(G_d)}\(\frac{1}{1-\gamma}\omega_{00}(u)-\frac{1}{1-\gamma}\omega_{00}(v)\)^2\leq 2^{d+1}.
\end{equation}
So for any $\gamma$, 
\begin{equation}
\sum_{(u,v)\in \overrightarrow{E}(G_d)}(\omega_{00}(u)-\omega_{00}(v))^2\leq (1-\gamma)^2 2^{d+1}.\label{eq:omega_00}
\end{equation}
By similar arguments, we have 
\begin{align}
\sum_{(u,v)\in \overrightarrow{E}(G_d)}(\omega_{01}(u)-\omega_{01}(v))^2&\leq \gamma^2 2^{d+1},\label{eq:omega_01}\\
\sum_{(u,v)\in \overrightarrow{E}(G_d)}(\omega_{10}(u)-\omega_{10}(v))^2&\leq (1-\gamma')^2 2^{d+1},\label{eq:omega_10}\\
\sum_{(u,v)\in \overrightarrow{E}(G_d)}(\omega_{11}(u)-\omega_{11}(v))^2&\leq (\gamma')^2 2^{d+1}.\label{eq:omega_11}\\
\end{align}

Thus, combining \eqref{eq:omega_00}, \eqref{eq:omega_01}, \eqref{eq:omega_10} and \eqref{eq:omega_11} with \eqref{eq:omegaA}:
\begin{equation}
\norm{\omega A}^2 \leq (1-\gamma)^22^{d+1}+\gamma^22^{d+1}+(1-\gamma')^22^{d+1}+(\gamma')^22^{d+1}.
\end{equation}

We note that if $r$ and $l$ are in the same
component as $s$ and $t$, then by the harmonic property, we must have
$\gamma,\gamma'\in [0,1]$, and so
$(1-\gamma)^2+\gamma^2+(1-\gamma')^2+(\gamma')^2\leq 2$, and so
$\norm{\omega A}^2\leq 2^{d+2}$. 

Otherwise, suppose $r$ is in a different component than $s$ and $t$. In
that case, $\omega$  can take any constant value on that component and
still be an approximate negative witness,  so we choose
$\tilde{\omega}=0$ on the component containing $r$ and
$\tilde{\omega}=\omega$ everywhere else. In that case, $\gamma'(\tilde{\omega})=0$,  so using arguments similar to those above, we have 
$\norm{\tilde\omega A}^2\leq 2^{d+2}$. Finally, since $\omega$ is an optimal approximate negative 
witness, $\norm{\omega
A}^2\leq \norm{\tilde\omega A}^2.$ A similar argument
works for the case when $l$ is in a different component.


The proof that $\tilde{w}_+(x)\leq 2^{d+1}$ for all $x\in\{0,1\}^{2^{2d}}$ is similar, so we omit the details.  
\end{proof}

\begin{proof}[Proof of Theorem \ref{thm:est}]
Let $x\in\{0,1\}^{2^{2d}}$.
By Theorem \ref{thm:span-est}, since ${\cal C}_B(x)={2}w_-(x)$ for all 0-instances, and ${\cal C}_A(x)=\frac{1}{2}w_+(x)$ for all 1-instances, we can 
estimate ${\cal C}_A(x)$ in query complexity 
\begin{equation}
\widetilde{O}\(\frac{1}{\eps^{2/3}}\sqrt{{\cal C}_A(x)}\widetilde{W}_-^{1/2}\)=\widetilde{O}\(\frac{1}{\eps^{2/3}}\sqrt{{\cal C}_A(x)2^d}\)=\widetilde{O}\(\frac{1}{\eps^{2/3}}\sqrt{{\cal C}_A(x)}N^{1/4}\), 
\end{equation}
since $N=2^{2d}$, and similarly, we can estimate ${\cal C}_B(x)$ in query complexity $\widetilde{O}\(\frac{1}{\eps^{2/3}}\sqrt{{\cal C}_B(x)}N^{1/4}\)$.
Extending to odd-depth trees is a straightforward exercise.
\end{proof}


\section{The Query Complexity of Winning a NAND-Tree}\label{sec:win_nand_tree}

Given an $A$-winnable input, suppose Player $A$ can access $x$ via queries
of the form: $\mathcal{O}_x:\ket{i,b}\mapsto \ket{i,b\oplus x_i}$. We
consider the number of queries needed by Player $A$ to make
decisions throughout the course of the game such that the final live
node has value 1 --- \emph{i.e}.\ Player $A$ wins the game --- with probability $\geq 2/3$. (In this section, we focus on $A$-winnable trees,
but the case of $B$-winnable trees is similar.)

\paragraph{Naive Strategy}
There is a quantum query algorithm \cite{Rei09,R01} 
that decides if a tree of depth ${2d}$ is winnable with bounded error $\epsilon$ in
$O(2^d\log \frac{1}{\epsilon})$ queries. Thus, if Player $A$
must make a decision at a node $v$ with children $v_0$ and $v_1$,
a naive strategy is to
evaluate the subtrees with roots $v_0$ and $v_1$ 
and move to one that evaluates to 1 (if they both evaluate to 0, then since Player $A$'s turns correspond to nodes labelled by $\vee$, $v$ also evaluates to 0, and the tree is not $A$-winnable).
By setting the error to $O(1/d)$ at each decision, 
this strategy succeeds with bounded error and costs:
\begin{equation}
 \sum_{\ell=0}^{d-1}2^{d-\ell}\log d = \sum_{r=1}^{d}2^r\log d \leq 2^{d+1}\log d=O(2^d\log d)=O(\sqrt{N}\log\log N).
\end{equation}

The naive strategy does not account for the fact that some subtrees may be
easier to win than others. For example, suppose one of the subtrees
rooted at $v_0$ or $v_1$ has all leaves labeled by 1, whereas the
other subtree has all leaves labeled by 0. In that case, the player
just needs to distinguish these two very disparate cases. More generally, one of the subtrees might have a
very small positive witness --- i.e., it is very winnable --- whereas
the other has a very large witness --- i.e., is not very winnable. In that case, it may save work later on if Player $A$ chooses the more winnable subtree.
Using the witness size estimation algorithm of \cite{IJ15},
we can create a strategy to quickly distinguish an easily winnable tree from a much
less easily winnable (possibly unwinnable) tree.

\paragraph{Strategy with Witness Size Estimation}

We now discuss a winning algorithm that uses witness size estimation. 
When the tree has small average choice complexity and the opponent
plays randomly, this strategy does better than the naive strategy, on average. The
new strategy will be to always choose the subtree with smaller average
choice complexity, unless the two subtrees are very close in average
choice complexity, in which case it doesn't matter which one we
choose. 

We estimate the average choice complexity of both
subtrees using Theorem \ref{thm:est}. Let $\mathtt{Est}(x)$ be the
algorithm from Thoerem \ref{thm:est} with $\eps=\frac{1}{3}$. We are
only concerned about which of two average choice complexities is
smaller, so we do not want to wait for both iterations of \texttt{Est}
to terminate.  Let $c$ be some polylogarithmic function in $N$ such
that \texttt{Est}$(x)$ always terminates after at most $c\sqrt{{\cal
C}_A(x)}N^{1/4}$, for all $x\in\{0,1\}^N$.  We define a subroutine,
$\texttt{Select}(x^0,x^1)$, that takes two instances, $x^0$ and $x^1$,
and outputs a bit $b$ such that ${\cal C}_A(x^b)\leq 2{\cal
C}_A(x^{\bar{b}})$, where $\bar{b}=b\oplus 1$. $\texttt{Select}$ works
as follows. It runs \texttt{Est}$(x^0)$ and $\texttt{Est}(x^1)$ in parallel. If one of these programs, say $\texttt{Est}(x^b)$,
outputs some estimate $w_b$, then it terminates the other program
after $c\sqrt{w_b}N^{1/4}$ steps. If only the algorithm running on $x^b$ has terminated after
this time, it outputs $b$. If both programs have terminated, it outputs a bit
$b$ such that $w_b\leq w_{\bar{b}}$.

\begin{lemma}\label{lemma:select}
Let $x^0,x^1\in\{0,1\}^N$ be \textsc{nand}-tree instances such that at least one of them is a 1-instance. Let $w_{\min}=\min\{{\cal C}_A(x^0),{\cal C}_A(x^1)\}$. Then
$\mathtt{Select}(x^0,x^1)$  terminates after
\begin{align}
\widetilde{O}\(N^{1/4}\sqrt{w_{\min}}\)
\end{align}
 queries to
$x^0$ and $x^1$ and outputs $b$ such that ${\cal C}_A(x^b)\leq
{2}{\cal C}_A(x^{\bar{b}})$ with bounded error.
\end{lemma}
\begin{proof}
Since at least one of $x^0$ and $x^1$ is a 1-instance, at least one of the programs will terminate. Suppose without loss of generality that $\texttt{Est}(x^0)$ is the first to terminate, outputting $w_0$. Then there are two possibilities: $\texttt{Est}(x^1)$ does not terminate after $c\sqrt{w_0}N^{1/4}$ steps, and \texttt{Select} outputs 0; or $\texttt{Est}(x^1)$ outputs $w_1$ before $c\sqrt{w}N^{1/4}$ steps have passed and \texttt{Select} outputs $b$ such that $w_b\leq w_{\bar{b}}$. 

Consider the first case. Suppose ${\cal C}_A(x^0)>2{\cal C}_A(x^1)$. Then $\texttt{Est}(x^1)$ must terminate after $c\sqrt{{\cal C}_A(x^1)}N^{1/4}\leq \frac{1}{\sqrt{2}}c\sqrt{{\cal C}_A(x^0)}N^{1/4}$ steps. By assumption, we have $|w_0-{\cal C}_A(x^0)|\leq \eps {\cal C}_A(x)$, so $w_0\geq (1-\eps){\cal C}_A(x^0)=\frac{2}{3}{\cal C}_A(x^0)$. Thus $\texttt{Est}(x^1)$ must terminate after $\frac{1}{\sqrt{2}}c\sqrt{\frac{3}{2}w_0}N^{1/4}<c\sqrt{w_0}N^{1/4}$ steps, which is a contradiction. Thus, ${\cal C}_A(x^0)\leq 2{\cal C}_A(x^1)$, so outputting 0 is correct. Furthermore, since we terminate after $c\sqrt{w_0}N^{1/4}=\widetilde{O}(\sqrt{{\cal C}_A(x^0)}N^{1/4})$ steps, and ${\cal C}_A(x^0)=O({\cal C}_A(x^1))$, the running time is at most $\widetilde{O}\(N^{1/4}\sqrt{w_{\min}}\)$.

We now consider the second case, in which both programs output estimates $w_0$ and $w_1$, such that $|w_b-{\cal C}_A(x^b)|\leq \eps {\cal C}_A(x^b)$ for $b=0,1$. Suppose $w_b\leq w_{\bar b}$. We then have
\begin{equation}
\frac{{\cal C}_A(x^b)}{{\cal C}_A(x^{\bar b})}\leq \frac{{\cal C}_A(x^b)}{w_b}\frac{w_{\bar b}}{{\cal C}_A(x^{\bar{b}})}\leq \frac{1+\eps}{1-\eps} = \frac{4/3}{2/3}=2. 
\end{equation}
Thus ${\cal C}_A(x^b)\leq 2{\cal C}_A(x^{\bar b})$, as required. Furthermore, the running time of the algorithm is bounded by the running time of $\texttt{Est}(x^1)$, the second to terminate. We know that $\texttt{Est}(x^1)$ has running time at most $\widetilde{O}(\sqrt{{\cal C}_A(x^1)}N^{1/4})$ steps, and by assumption, $\texttt{Est}(x^1)$ terminated after less than $c\sqrt{w_0}N^{1/4}=\widetilde{O}(\sqrt{{\cal C}_A(x^0)}N^{1/4})$ steps, so the total running time is at most $\widetilde{O}\(N^{1/4}\sqrt{w_{\min}}\)$.
\end{proof}

\noindent We can thus prove the main theorem of this section:
\begin{theorem}\label{thm:winning}
Let $x\in\{0,1\}^N$ for $N=2^{2d}$
be an $A$-winnable input. At every node $v$ 
where Player $A$ makes a decision, let Player $A$
use the $\mathtt{Select}$ algorithm in the following way.
Let $v_0$ and $v_1$ be the two children of $v$, with
inputs to the respective subtrees of $v_0$ and $v_1$ given
by $x^0$ and $x^1$ respectively. Then Player $A$ moves
to $v_b$ where $b$ is the outcome that occurs a majority of
times when $\mathtt{Select}(x^0,x^1)$
is run $O(\log d)$ times. 
Then if Player $B$, at his decision nodes, chooses left
and right with equal probability, Player $A$ will
win the game with probability at least $2/3$, and 
will use
$\tilde{O}(N^{1/4}\sqrt{{\cal C}(x)})$ queries on average, where the average
is taken over the randomness of Player $B$'s choices. 
\end{theorem}

\begin{proof}
We first note that Player $A$ must make $d$ choices
over the course of the game. Thus we amplify Player $A$'s probability
of success by repeating $\mathtt{Select}$ at each decision node $O(\log d)$
times and taking the majority. Then the probability that Player
$A$ chooses the wrong direction at any node is $O(1/d)$. By doing
this we ensure that her 
probability of choosing the wrong direction over the course of the algorithm
is $<1/3$. From here on, we analyze the error free case.

Let $v_{\tau}$ be any node in the tree at distance $2k$ from the root, 
for $k\in\{0,\dots,d-1\}$, with children $v_{\tau 0}$ and $v_{\tau 1}$. 
The node $v_{\tau}$ is the root of an instance of $\textsc{nand}$-tree of depth $2(d-k)$. 
Because it is the root of an even-depth subtree, it is a node where
Player $A$ must make a decision. Player $A$ runs
$\mathtt{Select}(x^{\tau 0},x^{\tau 1})$, which
returns $b_1\in\{0,1\}$ such that (by Lemma \ref{lemma:select})
\begin{align}\label{eq:select_cond}
{\cal C}_A(x^{\tau b_1})\leq 2
{\cal C}_A(x^{\tau \bar b_1}). \end{align}
If $x^{\tau b}$ represents the bits labeling the leaves of an odd-depth subtree, the root of the subtree, $v_{\tau b}$ 
is a node where Player $B$ makes a choice. Because we assume Player $B$ chooses uniformly at random,
\begin{eqnarray}\label{eq:odd_comple}
{\cal C}_A(x^{\tau b}) &=& \min_{{\cal S}\in \mathscr{S}_A(x^{\tau b})}\mathbb{E}_{p\in{\cal S}}\left[\prod_{v\in\nu_A(p)}c(v) \right]\notag\\
&=& \frac{1}{2}\(\min_{{\cal S}\in \mathscr{S}_A(x^{\tau b0})}\mathbb{E}_{p\in{\cal S}}\left[\prod_{v\in\nu_A(p)}c(v) \right]+\min_{{\cal S}\in \mathscr{S}_A(x^{\tau b1})}\mathbb{E}_{p\in{\cal S}}\left[\prod_{v\in\nu_A(p)}c(v) \right]\)\notag\\
&=& \frac{1}{2}\({\cal C}_A(x^{\tau b0})+{\cal C}_A(x^{\tau b1})\).
\end{eqnarray}
Thus, \eqref{eq:select_cond} becomes
\begin{equation}
\frac{1}{2}\({\cal C}_A(x^{\tau b_10})+{\cal C}_A(x^{\tau b_11})\)\leq {\cal C}_A(x^{\tau \bar{b}_10})+{\cal C}_A(x^{\tau \bar{b}_11}).\label{eq:select_cond2}
\end{equation}

If $v_{\tau}$ is a fault, then $\max\{{\cal C}_A(x^{\tau b_1}),{\cal C}_A(x^{\tau \bar b_1})\}=\infty$, and 
so we must have ${\cal C}_A(x^{\tau b_1})<{\cal C}_A(x^{\tau \bar b_1})=\infty$. 
Then by \eqref{eq:CR} and \eqref{eq:odd_comple}, ${\cal C}_A(x^{\tau})={\cal C}_A(x^{\tau b_10})+{\cal C}(x^{\tau b_11})=2{\cal C}_A(x^{\tau b})$.

If $v_\tau$ is not a fault, by \eqref{eq:C}, we have 
\begin{eqnarray}
{\cal C}_A(x^{\tau}) &=& \frac{({\cal C}_A(x^{\tau b_10})+{\cal C}_A(x^{\tau b_11}))({\cal C}_A(x^{\tau \bar{b}_10})+{\cal C}_A(x^{\tau \bar{b}_11}))}{{\cal C}_A(x^{\tau {b}_10})+{\cal C}_A(x^{\tau b_11})+{\cal C}_A(x^{\tau \bar{b}_10})+{\cal C}_A(x^{\tau \bar{b}_11})}\notag\\
&\geq & \frac{\frac{1}{2}\({\cal C}_A(x^{\tau b_10})+{\cal C}_A(x^{\tau b_11})\)^2}{\frac{3}{2}\({\cal C}_A(x^{\tau b_10})+{\cal C}_A(x^{\tau b_11})\)}\quad\mbox{by \eqref{eq:select_cond2}}.
\end{eqnarray}

Thus, whether $v_{\tau}$ is a fault or not, we have:
\begin{equation}
{\cal C}_A(x^{\tau b})=\frac{1}{2}\({\cal C}_A(x^{\tau b_10})+{\cal C}_A(x^{\tau b_11})\) \leq \frac{3}{2}{\cal C}_A(x^{\tau}). \label{eq:avg-growth}
\end{equation}

Since $v_\tau$ is the root of a \textsc{nand}-tree of size $N=2^{2(d-k)}$, by Lemma \ref{lemma:select}, there exists $c\in \mathrm{poly}(\log N)$ such that running \texttt{Select} at node $v_\tau$ has query complexity at most
\begin{align}
c\left(2^{2(d-k)}\right)^{1/4} \sqrt{{\cal C}_A(x^{\tau b})}
\leq c2^{(d-k)/2}\sqrt{\frac{3}{2}{\cal C}_A(x^{\tau})}.
\end{align}

We now show, by induction on $k$, that the expected value of this expression, in $\tau$, is bounded from above by $c2^{d/2}\sqrt{\frac{3}{2}{\cal C}_A(x)}$. For the $k=0$ step this is clear. We next consider the inductive step.

 Let $\tau b_1b_2$, for some string $\tau\in\{0,1\}^{2k-2}$ be the sequence of choices made thus far, so $b_2$ is the last choice made by Player $B$, and $b_1$ the last choice made by Player $A$. Then the query complexity of $\mathtt{Select}$ for $v^{\tau b_1b_2}$, that is, Player $A$'s next choice, will be at most $c2^{(d-k)/2}\sqrt{\frac{3}{2}{\cal C}_A(x^{\tau b_1b_2})}$. This has expected value, in Player $B$'s uniform random choice $b_2$:
\begin{eqnarray}
&& \frac{1}{2}\(c2^{(d-k)/2}\sqrt{\frac{3}{2}{\cal C}_A(x^{\tau b_10})}+c2^{(d-k)/2}\sqrt{\frac{3}{2}{\cal C}_A(x^{\tau b_11})}\)\notag\\
&\leq & \frac{c}{2}\sqrt{3}2^{(d-k)/2}\sqrt{\frac{1}{2}\left({\cal C}_A(x^{\tau b_10})+{\cal C}_A(x^{\tau b_11})\right)}\quad \mbox{by Jensen's inequlaity,}\notag\\
&= & \frac{c}{2}\sqrt{3}2^{(d-k)/2}\sqrt{{\cal C}_A(x^{\tau b_1})}\quad\mbox{by \eqref{eq:odd_comple}}\nonumber\\
&\leq & \frac{c\sqrt{3}}{2\sqrt{2}}2^{(d-(k-1))/2}\sqrt{\frac{3}{2}{\cal C}_A(x^{\tau })} \quad\mbox{by \eqref{eq:avg-growth}}. 
\end{eqnarray}
By the induction hypothesis, this has expected value in $\tau$ at  most $c\frac{\sqrt{3}}{2} 2^{d/2}\sqrt{\frac{3}{2}{\cal C}_A(x)}$.
Using $\sqrt{3}/2 < 1$ completes the proof that every call to \texttt{Select} has expected query complexity at most $c2^{d/2}\sqrt{\frac{3}{2}{\cal C}_A(x)}$.

Summing over the expected query complexity at all $d$ levels, and including success probability
amplification, 
we have that the query complexity is
\begin{equation}
\widetilde{O}\left(2^{d/2}\sqrt{{\cal C}_A(x)}\right)=\widetilde{O}\left(N^{1/4}\sqrt{{\cal C}_A(x)})\right).\qedhere
\end{equation}
\end{proof}

\section{Open Problems}\label{sec:open_problems}

In this work, we have introduced average choice complexity, a measure of the difficulty of winning the two-player game associated with a \textsc{nand}-tree, which allowed us to generalize a superpolynomial quantum speedup of \cite{ZKH12}, and give a new quantum algorithm for winning the two-player \textsc{nand}-tree game. To accomplish this, we exploited a connection between \textsc{nand}-trees and $st$-connectivity problems on a family of graphs, and their dual graphs, which we believe may be of further interest. 

One interesting problem raised by this work concerns span programs in general:
 when can we view span programs as solving $st$-connectivity
problems? This could be particularly interesting for understanding
when span programs are time-efficient, since the time-complexity
analysis of $st$-connectivity span programs is straightforward (see
\cite[Section 5.3]{BR12} or \cite[Appendix B]{IJ15}).

Another motivation for considering the connection between span programs
and $st$-connectivity problems is the nice characterization of the
duality between 1-instances and 0-instances given by the duality of
the corresponding graphs. An important class of $st$-connectivity-related span programs are those arising from the learning graph framework, which provides a means of designing quantum algorithms that is much simpler and more intuitive than designing a general span program \cite{Bel11}. A limitation of this framework is its one-sidedness with respect to 1-certificates: whereas a learning graph algorithm is designed to detect 1-certificates, a framework capable of giving optimal quantum query algorithms for any decision problem would likely treat 0- and 1-inputs symmetrically.  The duality between 1- and 0-inputs in $st$-connectivity problems could give insights into how to extend the learning graph framework to a more powerful framework, without losing the intuition and relative simplicity.

\section{Acknowledgments}

The authors would like to thank Ashley Montanaro for helpful
discussions.

S.J.\ gratefully acknowledges funding provided by the Institute for Quantum Information and Matter, 
an NSF Physics Frontiers Center (NFS Grant PHY-1125565) with support of the Gordon and Betty Moore Foundation (GBMF-12500028).
S.K. acknowledges funds provided by the Department of Defense.

\bibliographystyle{alpha}
\bibliography{nandbib}
\appendix

\end{document}